\documentclass{article}

\usepackage[colorlinks,allcolors=blue]{hyperref}
\usepackage{amssymb,amsmath,amsthm,mathrsfs}
\usepackage[T1]{fontenc}
\usepackage{venturis2}
\usepackage{graphicx}
\usepackage{cite}

\def\B{\mathscr B}
\def\C{\mathbb C}
\def\d{\mathrm{d}}
\def\DD{\mathscr D}
\def\dom{\mathcal D}
\def\F{\mathcal F}
\def\g{\mathfrak g}
\def\H{\mathcal H}
\def\h{\mathfrak h}
\def\id{\mathrm{id}}
\def\K{\mathscr K}
\def\L{\mathscr L}
\def\ltwo{\mathop{\mathrm{L}^2}\nolimits}

\def\linf{\mathop{\mathrm{L}^\infty}\nolimits}
\def\N{\mathbb N}

\def\R{\mathbb R}
\def\S{\mathbb S}
\def\SS{\mathscr S}

\def\Tau{\mathcal T}
\def\U{\mathrm U}
\def\UU{\mathscr U}
\def\Z{\mathbb Z}

\def\e{\mathop{\mathrm{e}}\nolimits}

\def\im{\mathop{\mathrm{Im}}\nolimits}
\def\re{\mathop{\mathrm{Re}}\nolimits}

\def\tsum{\mathop{\textstyle\sum}\nolimits}
\DeclareMathOperator*{\slim}{s\hspace{0.1pt}-\hspace{0.1pt}lim}

\newtheorem{Theorem}{Theorem}[section]
\newtheorem{Remark}[Theorem]{Remark}
\newtheorem{Lemma}[Theorem]{Lemma}

\newtheorem{Proposition}[Theorem]{Proposition}
\newtheorem{Assumption}[Theorem]{Assumption}

\renewcommand{\theequation}{\arabic{section}.\arabic{equation}}

\begin{document}

%--------------------------------------------------------------------------------------
% Title
%--------------------------------------------------------------------------------------

\title{Decay estimates for unitary representations with applications to continuous-
and discrete-time models}

\author{S. Richard${}^{1}$\footnote{Supported by the grant \emph{Topological
invariants through scattering theory and noncommutative geometry} from Nagoya
University, and by JSPS Grant-in-Aid for scientific research C no 18K03328 \&
21K03292, and on leave of absence from Univ.~Lyon, Universit\'e Claude Bernard Lyon 1,
CNRS UMR 5208, Institut Camille Jordan, 43 blvd. du 11 novembre 1918, F-69622
Villeurbanne cedex, France.} and R. Tiedra de Aldecoa${}^{2}$\footnote{Partially
supported by the Chilean Fondecyt Grant 1210003.}}

\date{\small}
\maketitle
\vspace{-1cm}

\begin{quote}
\begin{itemize}
\item[1] Graduate school of mathematics, Nagoya University,
Chikusa-ku,  Nagoya 464-8602, Japan
\item[2] Facultad de Matem\'aticas, Pontificia Universidad Cat\'olica de Chile,\\
Av. Vicu\~na Mackenna 4860, Santiago, Chile
\item[] E-mail:  richard@math.nagoya-u.ac.jp, rtiedra@mat.uc.cl
\end{itemize}
\end{quote}

%--------------------------------------------------------------------------------------

\begin{abstract}
We present a new technique to obtain polynomial decay estimates for the matrix
coefficients of unitary operators. Our approach, based on commutator methods, applies
to nets of unitary operators, unitary representations of topological groups, and
unitary operators given by the evolution group of a self-adjoint operator or by powers
of a unitary operator. Our results are illustrated with a wide range of examples in
quantum mechanics and dynamical systems, as for instance Schrödinger operators, Dirac
operators, quantum waveguides, horocycle flows, adjacency matrices, Jacobi matrices,
quantum walks or skew products.
\end{abstract}


\smallskip


%--------------------------------------------------------------------------------------
\tableofcontents
%--------------------------------------------------------------------------------------

%--------------------------------------------------------------------------------------
\section{Introduction and main results}\label{section_intro}
\setcounter{equation}{0}
%--------------------------------------------------------------------------------------

In recent papers
\cite{CT_2016,RT19,Sim_2018,Tie_2012,Tie_2015,Tie_2015_2,Tie_2017,Tie_2017_2,Tie_2018},
it has been shown that one can combine certain tools from dynamical systems (averaging
along the dynamics, ergodic theorems) and quantum mechanics (commutator methods) to
determine spectral properties of various classes of continuous- and discrete-time
models. In particular, a new criterion for strong mixing has been put into evidence in
\cite{Tie_2015,RT19}. Formally, it reads as follows: Let $(U_j)$ be a family of
unitary operators in a Hilbert space $\H$, let $\ell_j$ be positive numbers such that
$\ell_j\to\infty$, let $A$ be a self-adjoint operator in $\H$, set
$D_j:=\tfrac1{\ell_j}[A,U_j]U_j^{-1}$ and assume that the strong limit
$D:=\slim_j D_j$ exists. Then
\begin{equation}\label{eq_strong}
\lim_j|\langle\varphi,U_j\psi\rangle_\H|=0
\quad\hbox{for all $\varphi\in\ker(D)^\perp$ and $\psi\in\H$.}
\end{equation}
The operator $D$ can be interpreted as a topological degree of the map $j\mapsto U_j$.
Indeed, if one considers $[A,\;\!\cdot\;\!]$ as a derivation along the set $(U_j)$,
then $D$ corresponds to a renormalised, operator-valued, winding number for the map
$j\mapsto U_j$ (the logarithmic derivative $\frac{\d z}z$ in the usual definition of
winding number is replaced by the ``logaritmic derivative" $[A,U_j]U_j^{-1}$
associated to $[A,\;\!\cdot\;\!]$). See
\cite{Fra_2000,Fra_2004,GLL_1991,KL_2020,Kar_2016,Kar_2018,RT19,Tie_2015_2,Tie_2017,Tie_2018}
for more details and examples.

In concrete situations, one usually seeks to get an explicit rate of decay in
estimates like \eqref{eq_strong} in order to quantify the time propagation of the wave
functions. This problem is a very broad and active field of research, with numerous
results in a variety of setups. Decay of correlations, local decay estimates,
pointwise decay estimates, $\mathrm{L}^p$ decay estimates, Strichartz estimates,
microlocal estimates, propagation estimates, Morawetz estimates,... all are families
of results related to this problem. In this paper, we pursue the study initiated in
\cite{Tie_2015,RT19} and determine conditions that guarantee a polynomial rate of
decay in \eqref{eq_strong}. Our results are general, in the sense that they are stated
first for general nets of unitary operators, then for unitary representations of
topological groups, and finally for unitary operators given by the evolution group
$(\e^{-itH})_{t\in\R}$ of a self-adjoint operator $H$ or by the powers
$(U^n)_{n\in\Z}$ of a unitary operator $U$. Moreover, they apply to a wide range
of models both in quantum mechanics and dynamical systems. And finally, they
generalise to some extent the results of \cite{GLS_2016} (see also \cite{LS_2015})
where the authors use commutator methods to establish abstract pointwise decay
estimates for certain classes of self-adjoint operators. We refer to
\cite{BT_1999,DFV_2014,Las_1996} for related results about pointwise decay estimates
for self-adjoint and unitary operators.

Let us give a more detailed description of our results. In Section \ref{sec_general},
we introduce our framework and determine sufficient conditions that guarantee a
polynomial decay estimate
$$
|\langle\varphi,U_j\psi\rangle_\H|\le\tfrac1{\ell_j^n}\;\!c_{\varphi,\psi},
\quad\ell_j>0,~n\in \N^*,
$$
with $c_{\varphi,\psi}$ a constant depending on $\varphi,\psi$ (and $n$) but not on
$j$. A first set of conditions on $D_j$ and $\varphi,\psi$ leads to this estimate in
the case $n=1$, while more restrictive sets of conditions lead to this estimate for
any fixed $n\ge1$ (Theorem \ref{thm_decay}). In addition, when the unitary operators
$U_j$ are given by a unitary representation $\UU$ of a topological group $X$ and the
scalars $\ell_j$ are given by a proper length function on $X$, then we provide
conditions ensuring that $D$ commutes with $\UU$ and that $\UU$ has no nontrivial
finite-dimensional unitary subrepresentation in $\ker(D)^\perp$ (Proposition
\ref{prop_add}).

These general results are then applied in Sections \ref{sec_self} \& \ref{sec_unit} to
the case of a representation of $\R$ given by an evolution group
$(\e^{-itH})_{t\in\R}$ with self-adjoint generator $H$ and to the case of a
representation of $\Z$ given by the powers $(U^n)_{n\in\Z}$ of a unitary operator $U$.
In the former case, the operators $D_j$ can be written as Cesaro means
$$
D_t:=\tfrac1t\int_0^t\d\tau\,\e^{-i\tau H}(H+i)^{-1}[iH,A](H-i)^{-1}\e^{i\tau H},
\quad t>0,
$$
and our new results are the fact that $D$ is decomposable in the spectral
representation of $H$ and criteria for the continuity or absolute continuity of the
spectrum of $H$ in $\ker(D)^\perp$ (Lemma \ref{lemma_prop_H}). In the latter case, the
operators $D_j$ can be written as Cesaro means
$$
D_n=\tfrac1n\sum_{m=0}^{n-1}U^m([A,U]U^{-1})U^{-m},\quad n\in\N^*,
$$
and our new results are the fact that $D$ is decomposable in the spectral
representation of $U$ and criteria for the continuity or absolute continuity of the
spectrum of $U$ in $\ker(D)^\perp$ (Lemma \ref{lemma_prop_U}). Furthermore, in
Propositions \ref{prop_f(H)} \& \ref{prop_gamma(U)} we pay a special attention to the
particular cases $[iH,A]=f(H)$ and $[A,U]=\gamma(U)$ (with $f$ and $\gamma$ functions)
which are important for applications.

In Section \ref{sec_app}, we illustrate these abstract results with numerous examples.
For some of them, the decay estimates we obtain are known, while for others they are
new. A similar dichotomy holds for the spectral results we obtain when we deal with
representations admitting a self-adjoint or a unitary generator. However, the most
striking feature of our approach does not really rely on any new result for a given
example, but on its broad applicability. The whole variety of examples introduced in
Section \ref{sec_app} is conveniently covered with the same philosophy and toolkit.

Since Section \ref{sec_app} contains a detailed presentation of each example, we just
highlight here a few noticeable facts. First, we note that several important models of
quantum mechanics and dynamical systems are discussed in Section \ref{sec_app}. This
is for example the case of Schrödinger operators, Dirac operators, quantum waveguides,
horocycle flows, adjacency matrices, Jacobi matrices, quantum walks and skew
products. Next, as mentioned at the beginning of the introduction, the operator $D$
can sometimes be interpreted as a topological degree. This occurs for instance in the
case of skew products, see Section \ref{sec_skew}. In other instances, the operator
$D$ can be expressed in terms of the square of an asymptotic velocity operator (a
kinetic energy) for the unitary group under study. This occurs for instance in the
case of quantum walks on $\Z$, see Section \ref{sec_Z}. Finally, in Section
\ref{sec_reg} we discuss the case of the left regular representation of a
$\sigma$-compact locally compact Hausdorff group $X$ with left Haar measure $\mu$ and
proper length function $\ell$. In that case, we obtain for any net $(x_j)$ in $X$ with
$x_j\to\infty$ and suitable $\varphi,\psi\in\ltwo(X,\mu)$ the decay estimate
$$
|\langle\varphi,\UU(x_j)\psi\rangle_\H|\le\tfrac1{\ell(x_j)}\;\!c_{\varphi,\psi}.
$$
This estimate is similar to others in this paper, but with the interesting difference
that in general the representation $\UU$ doesn't have either a self-adjoint generator
or a unitary generator. It thus illustrates once again the fact that our approach is
general, and not only applicable to families of unitary operators admiting a
self-adjoint generator or unitary generator.

Finally, in Appendices \ref{sec_comm} \& \ref{sec_RAGE}, we collect some technical
results on commutators and regularity classes and on RAGE-type theorems for unitary
operators.

\bigskip
\noindent
{\bf Notations:} $\N:=\{0,1,2,\ldots\}$ is the set of natural numbers,
$\N^*:=\N\setminus\{0\}$, $\R_+:=(0,\infty)$, $\S^1$ the complex unit circle,
$\U(n)$ the group of $n\times n$ unitary matrices, and
$\langle\cdot\rangle:=\sqrt{1+|\cdot|^2}$. Given a Hilbert space $\H$, we write
$\|\cdot\|_\H$ for its norm, $\langle\cdot,\cdot\rangle_\H$ for its scalar product
(linear in the first argument), and $\U(\H)$ for the set of unitary operators on $\H$.
Given two Hilbert spaces $\H_1,\H_2$, we write $\B(\H_1,\H_2)$ (resp. $\K(\H_1,\H_2)$)
for the set of bounded (resp. compact) operators from $\H_1$ to $\H_2$. We also write
$\|\cdot\|_{\B(\H_1,\H_2)}$ for the norm of $\B(\H_1,\H_2)$, and use the shorthand
notations $\B(\H_1):=\B(\H_1,\H_1)$ and $\K(\H_1):=\K(\H_1,\H_1)$.

%--------------------------------------------------------------------------------------
\section{Decay estimates for unitary representations}\label{sec_decay}
\setcounter{equation}{0}
%--------------------------------------------------------------------------------------

%--------------------------------------------------------------------------------------
\subsection{General unitary representations}\label{sec_general}
%--------------------------------------------------------------------------------------

We start with a general theorem on decay estimates for the matrix coefficients of
unitary operators $U_j$ in a Hilbert space $\H$. In the proof, we use standard results
about commutators of operators recalled in Appendix \ref{sec_comm}.

\begin{Theorem}[Decay estimates]\label{thm_decay}
Let $(U_j)_{j\in J}$ be a net in $\U(\H)$, let $(\ell_j)_{j\in J}\subset[0,\infty)$
satisfy $\ell_j\to\infty$, assume there exists a self-adjoint operator $A$ in $\H$
such that $U_j\in C^1(A)$ for each $j\in J$, and suppose that the strong limit
$$
D:=\slim_jD_j\quad\hbox{with}\quad D_j:=\tfrac1{\ell_j}[A,U_j]U_j^{-1}
$$
exists. Then
\begin{enumerate}
\item[(a)] For each $\varphi=D\widetilde\varphi\in D\dom(A)$ and $\psi\in\dom(A)$ there
exists a constant $c_{\varphi,\psi}\ge0$ such that
$$
|\langle\varphi,U_j\psi\rangle_\H|
\le\|(D-D_j)\widetilde\varphi\|_\H\;\!\|\psi\|_\H
+\tfrac1{\ell_j}\;\!c_{\varphi,\psi},\quad\ell_j>0.
$$
In particular, $\lim_j\langle\xi,U_j\zeta\rangle_\H=0$ for all $\xi\in\ker(D)^\perp$
and $\zeta\in\H$.
\item[(b)] Assume that $D=D_j$ for all $j\in J$. Then for each $\varphi\in D\dom(A)$
and $\psi\in\dom(A)$ there exists a constant $c_{\varphi,\psi}\ge0$ such that
$$
|\langle\varphi,U_j\psi\rangle_\H|\le\tfrac1{\ell_j}\;\!c_{\varphi,\psi}
,\quad\ell_j>0.
$$
\item[(c)] Assume that $D=D_j$ for all $j\in J$, that $D\in C^1(A)$, and that
$[A,D]=DB$ with $B\in C^{(n-1)}(A)$ ($n\in\N^*$) and $[D,B]=0$. Then for each
$\varphi\in D^n\dom(A^n)$ and $\psi\in\dom(A^n)$ there exists a constant
$c_{\varphi,\psi}\ge0$ such that
$$
|\langle\varphi,U_j\psi\rangle_\H|\le\tfrac1{\ell_j^n}\;\!c_{\varphi,\psi},
\quad\ell_j>0.
$$
\end{enumerate}
\end{Theorem}

\begin{proof}
(a) Take $\varphi=D\widetilde\varphi\in D\dom(A)$, $\psi\in\dom(A)$, and $j\in J$ such
that $\ell_j>0$. Then we have
\begin{align*}
|\langle\varphi,U_j\psi\rangle_\H|
&=\big|\langle(D-D_j)\widetilde\varphi,U_j\psi\rangle_\H
+\langle D_j\widetilde\varphi,U_j\psi\rangle_\H\big|\\
&\le\|(D-D_j)\widetilde\varphi\|_\H\!\;\|\psi\|_\H+\tfrac1{\ell_j}
\big|\langle[A,U_j]U_j^{-1}\widetilde\varphi,U_j\psi\rangle_\H\big|\\
&\le\|(D-D_j)\widetilde\varphi\|_\H\;\!\|\psi\|_\H
+\tfrac1{\ell_j}\big|\langle A\widetilde\varphi,U_j\psi\rangle_\H\big|
+\tfrac1{\ell_j}\big|\langle\widetilde\varphi,U_jA\psi\rangle_\H\big|\\
&\le\|(D-D_j)\widetilde\varphi\|_\H\;\!\|\psi\|_\H
+\tfrac1{\ell_j}\;\!c_{\varphi,\psi}
\end{align*}
with
$
c_{\varphi,\psi}
:=\|A\widetilde\varphi\|_\H\|\psi\|_\H+\|\widetilde\varphi\|_\H\|A\psi\|_\H
$.
This proves the first part of the claim. Since $D=\slim_jD_j$ and $\ell_j\to\infty$,
we infer that $\lim_j\langle\varphi,U_j\psi\rangle=0$, and thus the second part of the
claim follows by the density of $D\dom(A)$ in $\overline{D\H}=\ker(D)^\perp$ and the
density of $\dom(A)$ in $\H$.

(b) The claim is a direct consequence of point (a) in the case $D=D_j$ for all
$j\in J$.

(c) We prove the claim by induction on $n\in\N^*$. For $n=1$, the claim is true due to
point (b). For $n-1\ge1$, we make the induction hypothesis that the claim is true. For
$n$, we take $\varphi\in D^n\dom(A^n)$, $\psi\in\dom(A^n)$ and $j\in J$ such that
$\ell_j>0$. Then, since $D=D_j$ and $\varphi=D\widetilde\varphi$ with
$\widetilde\varphi\in D^{n-1}\dom(A^n)$, we get
\begin{equation}\label{eq_one}
\langle\varphi,U_j\psi\rangle_\H
=\langle D_j\widetilde\varphi,U_j\psi\rangle_\H
=\tfrac1{\ell_j}\langle A\widetilde\varphi,U_j\psi\rangle_\H
-\tfrac1{\ell_j}\langle\widetilde\varphi,U_jA\psi\rangle_\H.
\end{equation}
The induction hypothesis applies to the second term in \eqref{eq_one} since
$\widetilde\varphi\in D^{n-1}\dom(A^n)\subset D^{n-1}\dom(A^{n-1})$ and
$A\psi\in\dom(A^{n-1})$. So there exists $c_{\widetilde\varphi,A\psi}\ge0$ such that
\begin{equation}\label{eq_two}
\big|\langle\widetilde\varphi,U_jA\psi\rangle_\H\big|
\le\tfrac1{\ell_j^{n-1}}\;\!c_{\widetilde\varphi,A\psi}.
\end{equation}
For the first term in \eqref{eq_one}, we have
$\widetilde\varphi=D^{n-1}\tilde{\tilde\varphi}$ with
$\tilde{\tilde\varphi}\in\dom(A^n)$. So, using the relations $[A,D]=DB$ and $[D,B]=0$,
we get that
\begin{align*}
A\widetilde\varphi
&=\big(D^{n-1}A+[A,D^{n-1}]\big)\tilde{\tilde\varphi}\\
&=\big(D^{n-1}A+\tsum_{m=0}^{n-2}D^{n-2-m}[A,D]D^m\big)\tilde{\tilde\varphi}\\
&=D^{n-1}(A+(n-1)B)\tilde{\tilde\varphi}
\end{align*}
with $(A+(n-1)B)\tilde{\tilde\varphi}\in\dom(A^{n-1})$ due to the inclusions
$\tilde{\tilde\varphi}\in\dom(A^n)$ and $B\in C^{(n-1)}(A)$. Therefore
$A\widetilde\varphi\in D^{n-1}\dom(A^{n-1})$ and
$\psi\in\dom(A^n)\subset\dom(A^{n-1})$, and we infer from the induction hypothesis
that there exists $c_{A\widetilde\varphi,\psi}\ge0$ such that
\begin{equation}\label{eq_three}
\big|\langle A\widetilde\varphi,U_j\psi\rangle_\H\big|
\le\tfrac1{\ell_j^{n-1}}\;\!c_{A\widetilde\varphi,\psi}.
\end{equation}
Finally, combining \eqref{eq_one}, \eqref{eq_two} and \eqref{eq_three}, we obtain that
$$
|\langle\varphi,U_j\psi\rangle_\H|
\le\tfrac1{\ell_j}\tfrac1{\ell_j^{n-1}}\;\!c_{A\widetilde\varphi,\psi}
+\tfrac1{\ell_j}\tfrac1{\ell_j^{n-1}}\;\!c_{\widetilde\varphi,A\psi}
=\tfrac1{\ell_j^n}\;\!c_{\varphi,\psi}
$$
with $c_{\varphi,\psi}:=c_{A\widetilde\varphi,\psi}+c_{\widetilde\varphi,A\psi}$.
\end{proof}

\begin{Remark}
(a) If the operators $U_j$ are given by a unitary representation, then the property of
Theorem \ref{thm_decay}(a) $\lim_j\langle\xi,U_j\zeta\rangle_\H=0$ for
$\xi\in\ker(D)^\perp$ and $\zeta\in\H$ amounts to a strong mixing property of the
unitary representation in $\ker(D)^\perp$. See \cite{RT19} for more information on
this point.

(b) The set $D^n\dom(A^n)$ in Theorem \ref{thm_decay}(c) is always dense in
$\ker(D)^\perp$, independently of the value of $n\in\N^*$. Indeed, since $\dom(A^n)$
is dense in $\H$, we have that $D^n\dom(A^n)$ is dense in
$\overline{D^n\H}=\ker(D^n)^\perp$. But $D$ is self-adjoint. So $\ker(D^n)=\ker(D)$,
and thus $D^n\dom(A^n)$ is dense in $\ker(D)^\perp$.

(c) Sometimes the unitary operators $U_j$ are given by the evolution group of a
self-adjoint operator $H$, namely, $(U_j)_{j\in J}=(\e^{-itH})_{t>0}$. In this
situation, a convenient operator $A$ in Theorem \ref{thm_decay} is often of the form
$\widetilde A=(H+i)^{-1}A(H-i)^{-1}$ with $A$ some self-adjoint operator such that
$(H-i)^{-1}\in C^1(A)$. In that case, the estimates of Theorem \ref{thm_decay} hold
for vectors $\varphi\in D^n\dom((\widetilde A)^n)$ and
$\psi\in\dom((\widetilde A)^n)$. However, since $A$ is simpler than $\widetilde A$ and
since $\dom(A^n)\subset\dom((\widetilde A)^n)$, in concrete examples we will only
present the estimates for vectors $\varphi\in D^n\dom(A^n)$ and $\psi\in\dom(A^n)$ for
the sake of simplicity (see Sections \ref{sec_Schrod}, \ref{sec_guides},
\ref{sec_horo}, \ref{sec_graphs} and \ref{sec_Jacobi}).
\end{Remark}

In the sequel, we assume that the unitary operators $U_j$ are given by a unitary
representation $\UU:X\to\U(\H)$ of a topological group $X$. We also assume that the
scalars $\ell_j$ are given by a proper length function on $X$, that is, a function
$\ell:X\to[0,\infty)$ satisfying the following properties (with $e$ the identity of
$X$):
\begin{enumerate}
\item[(L1)] $\ell(e)=0$,
\item[(L2)] $\ell(x^{-1})=\ell(x)$ for all $x\in X$,
\item[(L3)] $\ell(xy)\le\ell(x)+\ell(y)$ for all $x,y\in X$,
\item[(L4)] if $K\subset[0,\infty)$ is compact, then $\ell^{-1}(K)\subset X$ is
relatively compact.
\end{enumerate}
Finally, we recall that a net $(x_j)_{j\in J}$ in a topological space $X$ diverges to
infinity, with notation $x_j\to\infty$, if $(x_j)_{j\in J}$ has no limit point in $X$.
This implies that for each compact set $K\subset X$, there exists $j_K\in J$ such that
$x_j\notin K$ for $j\ge j_K$. In particular, $X$ is not compact.

In this situation, the existence of the strong limit $D$ leads to additional
properties of the unitary operators given by $\UU$. Namely, $\UU$ has no nontrivial
finite-dimensional unitary subrepresentation in $\ker(D)^\perp$, and the operator $D$
commutes with $\UU:$

\begin{Proposition}\label{prop_add}
Let $X$ be a topological group equipped with a proper length function $\ell$, let
$\UU:X\to\U(\H)$ be a unitary representation of $X$, let $(x_j)_{j\in J}$ be a net in
$X$ with $x_j\to\infty$, assume there exists a self-adjoint operator $A$ in $\H$ such
that $\UU(x_j)\in C^1(A)$ for each $j\in J$, and suppose that the strong limit
$$
D:=\slim_jD_j
\quad\hbox{with}\quad
D_j:=\tfrac1{\ell(x_j)}[A,\UU(x_j)]\UU(x_j)^{-1}
$$
exists. Then
\begin{enumerate}
\item[(a)] $\UU$ has no nontrivial finite-dimensional unitary subrepresentation in
$\ker(D)^\perp$.
\item[(b)] Let $x\in X$, assume that $\UU(x)\in C^1(A)$, and suppose that the strong
limit
$$
\widetilde D:=\slim_j\widetilde D_j
\quad\hbox{with}\quad
\widetilde D_j:=\tfrac1{\ell(x^{-1}x_j)}[A,\UU(x^{-1}x_j)]\UU(x^{-1}x_j)^{-1}
$$
exists and satisfies $D=\widetilde D$. Then $[D,\UU(x)]=0$.
\end{enumerate}
\end{Proposition}

\begin{proof}
(a) The claim follows from Theorem \ref{thm_decay}(a) and the fact that matrix
coefficients of finite-dimensional unitary representations of a group do not vanish at
infinity (see for instance \cite[Rem.~2.15(iii)]{BM00}).

(b) First, note that the commutator $[A,\UU(x^{-1}x_j)]$ in the expression for
$\widetilde D$ is well-defined for each $j\in J$ because
$\UU(x^{-1}x_j)=\UU(x)^{-1}\UU(x_j)$ with $\UU(x)\in C^1(A)$ and $\UU(x_j)\in C^1(A)$.
Next, we have
\begin{align*}
D\UU(x)
&=\slim_j\tfrac1{\ell(x_j)}[A,\UU(x)\UU(x^{-1}x_j)]\UU(x^{-1}x_j)^{-1}\\
&=\slim_j\tfrac1{\ell(x_j)}[A,\UU(x)]\UU(x^{-1}x_j)\UU(x^{-1}x_j)^{-1}\\
&\quad+\UU(x)\cdot\slim_j\tfrac1{\ell(x_j)}[A,\UU(x^{-1}x_j)]\UU(x^{-1}x_j)^{-1}\\
&=\lim_j\tfrac1{\ell(x_j)}\cdot[A,\UU(x)]
+\UU(x)\cdot\slim_j\tfrac1{\ell(x_j)}[A,\UU(x^{-1}x_j)]\UU(x^{-1}x_j)^{-1}
\end{align*}
with the first term vanishing because $\lim_j\tfrac1{\ell(x_j)}=0$ and with the second
term satisfying
\begin{align*}
&\UU(x)\cdot\slim_j\tfrac1{\ell(x_j)}[A,\UU(x^{-1}x_j)]\UU(x^{-1}x_j)^{-1}\\
&=\UU(x)\widetilde D+\UU(x)\cdot\slim_j\left(\tfrac1{\ell(x_j)}
-\tfrac1{\ell(x^{-1}x_j)}\right)[A,\UU(x^{-1}x_j)]\UU(x^{-1}x_j)^{-1}.
\end{align*}
Thus, to conclude the proof, it is sufficient to show that
\begin{equation}\label{eq_zero}
\slim_j\left(\tfrac1{\ell(x_j)}-\tfrac1{\ell(x^{-1}x_j)}\right)
[A,\UU(x^{-1}x_j)]\UU(x^{-1}x_j)^{-1}=0.
\end{equation}
Let $\varphi\in\H$. Then we have
\begin{align*}
&\lim_j\left\|\left(\tfrac1{\ell(x_j)}-\tfrac1{\ell(x^{-1}x_j)}\right)
[A,\UU(x^{-1}x_j)]\UU(x^{-1}x_j)^{-1}\varphi\right\|_\H\\
&=\lim_j\left|\tfrac{\ell(x^{-1}x_j)-\ell(x_j)}{\ell(x_j)}\right|
\cdot\left\|\tfrac1{\ell(x^{-1}x_j)}[A,\UU(x^{-1}x_j)]\UU(x^{-1}x_j)^{-1}
\varphi\right\|_\H\\
&\le\big(\|\widetilde D\varphi\|_\H+1\big)\lim_j
\left|\tfrac{\ell(x^{-1}x_j)-\ell(x_j)}{\ell(x_j)}\right|.
\end{align*}
Since $|\ell(x^{-1}x_j)-\ell(x_j)|\le\ell(x)$ due to the triangle inequality for
$\ell$, we infer that
$$
\lim_j\left\|\left(\tfrac1{\ell(x_j)}-\tfrac1{\ell(x^{-1}x_j)}\right)
[A,\UU(x^{-1}x_j)]\UU(x^{-1}x_j)^{-1}\varphi\right\|_\H
\le\big(\|\widetilde D\varphi\|_\H+1\big)\ell(x)\lim_j\tfrac1{\ell(x_j)}=0,
$$
which proves \eqref{eq_zero}.
\end{proof}

%--------------------------------------------------------------------------------------
\subsection{Unitary representations with self-adjoint generator}\label{sec_self}
%--------------------------------------------------------------------------------------

In this section, we consider the important case where the representation is a strongly
continuous unitary representation $\UU:\R\to\U(\H)$ of the additive group $\R$. In
such a case, Stone's theorem implies the existence of a self-adjoint operator $H$ in
$\H$ such that $\UU(t)=\e^{-itH}$ for each $t\in\R$. One could also consider the
higher-dimensional case of a strongly continuous unitary representation of the
additive group $\R^d$ for $d\ge1$. But we refrained from doing it for the sake of
simplicity.

We use the notation $P_{\rm p}(H)$ (resp. $P_{\rm c}(H)$, $P_{\rm ac}(H)$) for the
projection onto the pure point (resp. continuous, absolutely continuous) subspace
$\H_{\rm p}(H)$ (resp. $\H_{\rm c}(H)$, $\H_{\rm ac}(H)$) of $H$, $E^H(\cdot)$ for the
spectral projections of $H$, and $\chi_{\mathcal B}$ for the characteristic function
of a Borel set $\mathcal B\subset\R$.

\begin{Lemma}[Properties of $D$]\label{lemma_prop_H}
Let $H$ and $A$ be self-adjoint operators in a Hilbert space $\H$ with
$(H-i)^{-1}\in C^1(A)$, and let
$$
D_t:=\tfrac1t\int_0^t\d\tau\,\e^{-i\tau H}(H+i)^{-1}[iH,A](H-i)^{-1}\e^{i\tau H},
\quad t>0.
$$
Then
\begin{enumerate}
\item[(a)] $\slim_{t\to\infty}D_tP_{\rm p}(H)=0$.
\item[(b)] If there exists $B\in\B(\H)$ such that $(H+i)^{-1}([iH,A]-B)(H-i)^{-1}\in\K(\H)$
and
$$
\slim_{t\to\infty}\tfrac1t\int_0^t\d\tau\,
\e^{-i\tau H}(H+i)^{-1}B(H-i)^{-1}\e^{i\tau H}P_{\rm c}(H)~~\hbox{exists,}
$$
then
$$
\slim_{n\to\infty}D_t
=\slim_{t\to\infty}\tfrac1t\int_0^t\d\tau\,
\e^{-i\tau H}(H+i)^{-1}B(H-i)^{-1}\e^{i\tau H}P_{\rm c}(H).
$$
\end{enumerate}
Furthermore, if $D:=\slim_{t\to\infty}D_t$ exists, then
\begin{enumerate}
\item[(c)] $[D,\e^{isH}]=0$ for all $s\in\R$. In particular, $D$ is decomposable in the
spectral representation of $H$.
\item[(d)] $D=DP_{\rm c}(H)$. In particular, $H|_{\ker(D)^\perp}$ has purely
continuous spectrum.
\item[(e)] If $D\dom(A)\subset\dom(A)$ and
$\int_1^\infty\d t\;\!\|(D-D_t)\varphi\|^2_\H<\infty$ for all $\varphi\in\dom(A)$,
then $H|_{\ker(D)^\perp}$ has purely a.c. spectrum.
\end{enumerate}
\end{Lemma}

Point (c) implies that $\ker(D)^\perp$ is a reducing subspace for $H$. Therefore the
operator $H|_{\ker(D)^\perp}$ in points (d)-(e) is a well-defined self-adjoint
operator (see \cite[Thm.~7.28]{Wei_1980}).

\begin{proof}
(a) Let $\varphi\in\H$. Then $P_{\rm p}(H)\varphi=\sum_{j\ge1}\alpha_j\varphi_j$ with
$(\varphi_j)_{j\ge1}$ an orthonormal basis of $\H_{\rm p}(H)$, $\alpha_j\in\C$, and
$H\varphi_j=\lambda_j\varphi_j$ for some $\lambda_j\in\R$. Thus we obtain
\begin{equation}\label{eq_ex_H}
\slim_{t\to\infty}D_tP_{\rm p}(H)\varphi
=\slim_{t\to\infty}\sum_{j\ge1}\alpha_j
\left(\tfrac1t\int_0^t\d\tau\,\e^{-i\tau(H-\lambda_j)}\right)
(H+i)^{-1}[iH,A](H-i)^{-1}\varphi_j.
\end{equation}
Now
$
\|\tfrac1t\int_0^t\d\tau\,\e^{-i\tau(H-\lambda_j)}\|_{\B(\H)}\le1
$
for all $t>0$, and
$$
\slim_{t\to\infty}\tfrac1t\int_0^t\d\tau\,\e^{-i\tau(H-\lambda_j)}=E^H(\{\lambda_j\})
$$
due to von Neumann's mean ergodic theorem. Therefore we can exchange the limit and the
sum in \eqref{eq_ex_H} to get
\begin{align*}
\slim_{t\to\infty}D_tP_{\rm p}(H)\varphi
&=\sum_{j\ge1}\alpha_jE^H(\{\lambda_j\})(H+i)^{-1}[iH,A](H-i)^{-1}\varphi_j\\
&=\sum_{j\ge1}\alpha_j\langle\lambda_j\rangle^{-2}
E^H(\{\lambda_j\})[iH,A]E^H(\{\lambda_j\})\varphi_j.
\end{align*}
But $E^H(\{\lambda_j\})[iH,A]E^H(\{\lambda_j\})=0$ for each $\lambda_j$ due to the
virial theorem for self-adjoint operators \cite[Prop.~7.2.10]{ABG_1996}. Thus we
obtain that $\slim_{t\to\infty}D_tP_{\rm p}(H)\varphi=0$, which proves the claim.

(b) Let $K:=(H+i)^{-1}([iH,A]-B)(H-i)^{-1}\in\K(\H)$. Then it follows from point (a)
that
\begin{align*}
\slim_{t\to\infty}D_t
&=\slim_{t\to\infty}\tfrac1t\int_0^t\d\tau\,
\e^{-i\tau H}(H+i)^{-1}B(H-i)^{-1}\e^{i\tau H}P_{\rm c}(H)\\
&\quad+\slim_{t\to\infty}\tfrac1t\int_0^t\d\tau\,\e^{-i\tau H}K\e^{i\tau H}P_{\rm c}(H).
\end{align*}
with
$
\slim_{t\to\infty}\tfrac1t\int_0^t\d\tau\,\e^{-i\tau H}K\e^{i\tau H}P_{\rm c}(H)=0
$
due to \cite[Thm.~5.9]{Tes_2014}.

(c) The claim follows from Proposition \ref{prop_add}(b) in the case of the additive
group $X=\R$ and the auxiliary operator
$$
\widetilde A\varphi:=(H+i)^{-1}A(H-i)^{-1}\varphi,\quad\varphi\in\dom(A).
$$
Indeed, we know from \cite[Cor.~2.7 \& Rem.~2.8]{RT19} that $\widetilde A$ is
essentially self-adjoint (with closure denoted by the same symbol) and that
$\e^{-itH}\in C^1(\widetilde A)$ with $[\widetilde A,\e^{-itH}]=tD_t\e^{-itH}$ for any
$t>0$. Therefore, if we take the proper length function $\ell:\R\to[0,\infty)$ given
by $\ell(t):=|t|$, the unitary representation $\UU:\R\to\U(\H)$ given by
$\UU(t):=\e^{-itH}$, and the net $(x_j)_{j\in J}=(t)_{t>0}$, then we get
$\UU(s)=\e^{-isH}\in C^1(\widetilde A)$ for all $s\in\R$ and
\begin{align*}
\widetilde D
=\slim_{t\to\infty}\tfrac1{t-s}[\widetilde A,\e^{-i(t-s)H}]\e^{i(t-s)H}
=\slim_{t\to\infty}\tfrac1t[\widetilde A,\e^{-itH}]\e^{itH}
=\slim_{t\to\infty}D_t
=D.
\end{align*}
So all the assumptions of Proposition \ref{prop_add}(b) are verified, and thus
$[D,\e^{isH}]=0$ for all $s\in\R$.

Finally, since $[D,\e^{isH}]=0$ for all $s\in\R$, we have
$D\chi_{\mathcal B}(H)=\chi_{\mathcal B}(H)D$ for each Borel set
$\mathcal B\subset\R$, and thus $D$ is decomposable in the spectral representation of
$H$ \cite[Thm.~7.2.3(b)]{BS_1987}.

(d) The equality $D=DP_{\rm c}(H)$ follows from point (a). As a consequence, we get
that $\ker(D)^\perp\subset\H_{\rm c}(H)$, and thus the operator $H|_{\ker(D)^\perp}$
has purely continuous spectrum.

(e) Take $\varphi\in\dom(A)$. Then
$\psi=D\varphi\in D\dom(\widetilde A)\cap\dom(\widetilde A)$, and it follows from
Theorem \ref{thm_decay}(a) that there exists $c_\psi\ge0$ such that
$$
\big|\langle\psi,\e^{-itH}\psi\rangle_\H\big|
\le\big\|(D-D_t)\varphi\big\|_\H\;\!\|\psi\|_\H+\tfrac1t\;\!c_\psi,\quad t>0.
$$
So, we infer from the assumption and Cauchy-Schwarz inequality that
$\int_1^\infty\d t\;\!|\langle\psi,\e^{-itH}\psi\rangle_\H|^2<\infty$, and thus that
$t\mapsto\langle\psi,\e^{-itH}\psi\rangle_\H$ belongs to $\ltwo(\R)$. Therefore,
Plancherel's theorem for the group $X=\R$ \cite[Thm.~4.26]{Fol16} implies that
$\psi\in D\dom(A)$ belongs to the a.c. subspace $\H_{\rm ac}(H)$ of $H$. Thus
$H|_{\ker(D)^\perp}$ has purely a.c. spectrum, since $D\dom(A)$ is dense in
$\overline{D\H}=\ker(D)^\perp$ and $\H_{\rm ac}(H)$ is closed in $\H$.
\end{proof}

In the next proposition, we consider the particular case where $[iH,A]=f(H)$ for some
Borel function $f:\R\to\R$. Our results in this case generalise the results of
\cite[Cor.~4.3-4.4]{Tie_2017}.

\begin{Proposition}[The case {$[iH,A]=f(H)$}]\label{prop_f(H)}
Let $H$ and $A$ be self-adjoint operators in a Hilbert space $\H$, assume that
$(H-i)^{-1}\in C^1(A)$ with $[iH,A]=f(H)$ for some Borel function $f:\R\to\R$, and set
$g:=f\langle\cdot\rangle^{-2}$. Then
\begin{enumerate}
\item[(a)] For each $\varphi\in g(H)\dom(A)$ and $\psi\in\dom(A)$ there exists a
constant $c_{\varphi,\psi}\ge0$ such that
$$
\big|\langle\varphi,\e^{-itH}\psi\rangle_\H\big|
\le\tfrac1t\;\!c_{\varphi,\psi},\quad t>0.
$$
\item[(b)] If $g(H)\dom(A)\subset\dom(A)$, then $H|_{\ker(f(H))^\perp}$ has purely
a.c. spectrum.
\item[(c)] Suppose that $f\in C^n(\R)$ ($n\in\N^*$) with
$g^{(k)}\in\ltwo(\R)\cap\linf(\R)$ for all $k=0,\dots,n$ and $g^{(n)}$ uniformly
continuous. Then for each $\varphi\in g(H)^n\dom(A^n)$ and $\psi\in\dom(A^n)$ there
exists a constant $c_{\varphi,\psi}\ge0$ such that
$$
\big|\langle\varphi,\e^{-itH}\psi\rangle_\H\big|
\le\tfrac1{t^n}\;\!c_{\varphi,\psi},\quad t>0.
$$
\end{enumerate}
\end{Proposition}

\begin{proof}
(a) We know from the proof of Lemma \ref{lemma_prop_H}(c) that the auxiliary operator
$$
\widetilde A\varphi=(H+i)^{-1}A(H-i)^{-1}\varphi,\quad\varphi\in\dom(A),
$$
is essentially self-adjoint (with closure denoted by the same symbol) and that
$\e^{-itH}\in C^1(\widetilde A)$ with $[\widetilde A,\e^{-itH}]=tD_t\e^{-itH}$ for any
$t>0$. So the strong limit $D$ exists and satisfies the equalities
\begin{equation}\label{eq_same}
D_t
=\tfrac1t[\widetilde A,\e^{-itH}]\e^{itH}
=\tfrac1t\int_0^t\d\tau\,\e^{-i\tau H}(H+i)^{-1}[iH,A](H-i)^{-1}\e^{i\tau H}
=g(H)
=D.
\end{equation}
Therefore the assumptions of Theorem \ref{thm_decay}(b) are satisfied for the net
$(U_j)_{j\in J}=(\e^{-itH})_{t>0}$, the set $(\ell_j)_{j\in J}=(t)_{t>0}$, the
operator $\widetilde A$, and $D=g(H)$. Thus for each
$\varphi\in g(H)\dom(A)\subset g(H)\dom(\widetilde A)$ and
$\psi\in\dom(A)\subset\dom(\widetilde A)$ there exists a constant
$c_{\varphi,\psi}\ge0$ such that
$$
\big|\langle\varphi,\e^{-itH}\psi\rangle_\H\big|
\le\tfrac1t\;\!c_{\varphi,\psi},\quad t>0.
$$

(b) The claim follows from Lemma \ref{lemma_prop_H}(e) because $D=g(H)=D_t$ for all
$t>0$ and
$$
\ker(g(H))
=\chi_{g^{-1}(\{0\})}(H)\H
=\chi_{f^{-1}(\{0\})}(H)\H
=\ker(f(H)).
$$

(c) Since $g$ is uniformly continuous and belongs to $\ltwo(\R)\cap\linf(\R)$, we have
\begin{equation}\label{eq_approx}
\lim_{\varepsilon\searrow0}\big\|g(H)-g^\varepsilon(H)\big\|_{\B(\H)}=0
\end{equation}
with
$$
g^\varepsilon(H):=\int_\R\d t\;\!(\F g)(t)\e^{-(\varepsilon t)^2}\e^{2\pi itH}
\quad\hbox{(strong or Bochner integral)}
$$
and $\F:\ltwo(\R)\to\ltwo(\R)$ the Fourier transform (see
\cite[Thm.~8.35(b)]{Fol_1999}). Using successively the fact that $D=g(H)$, equation
\eqref{eq_approx}, the inclusion $\e^{2\pi itH}\dom(A)\subset\dom(\widetilde A)$,
equation \eqref{eq_same}, the relation $2\pi it(\F g)(t)=(\F g')(t)$, and
\eqref{eq_approx} with $g$ replaced by $g'$, we get for $\varphi\in\dom(A)$ the
equalities
\begin{align*}
\langle\widetilde A\varphi,D\varphi\rangle_\H
-\langle\varphi,D\widetilde A\varphi\rangle_\H
&=\lim_{\varepsilon\searrow0}\int_\R\d t\,\overline{(\F g)(t)}\e^{-(\varepsilon t)^2}
\big(\langle\widetilde A\varphi,\e^{2\pi itH}\varphi\rangle_\H
-\langle\varphi,\e^{2\pi itH}\widetilde A\varphi\rangle_\H\big)\\
&=\lim_{\varepsilon\searrow0}\int_\R\d t\,\overline{(\F g)(t)}
\e^{-(\varepsilon t)^2}\langle\varphi,[\widetilde A,\e^{2\pi itH}]\varphi\rangle_\H\\
&=\lim_{\varepsilon\searrow0}\int_\R\d t\,\overline{(\F g)(t)}
\e^{-(\varepsilon t)^2}\langle\varphi,(-2\pi t)D\e^{2\pi itH}\varphi\rangle_\H\\
&=\lim_{\varepsilon\searrow0}\int_\R\d\tau\,\overline{i(\F g')(t)}
\e^{-(\varepsilon t)^2}\langle\varphi,D\e^{2\pi itH}\varphi\rangle_\H\\
&=\langle\varphi,iDg'(H)\varphi\rangle_\H.
\end{align*}
Since $g'(H)\in\B(\H)$ due to the inclusion $g'\in\linf(\R)$, we infer that
$D\in C^1(\widetilde A)$ with $[\widetilde A,D]=iDg'(H)$. Thus, in the present case,
the operator $B\in\B(\H)$ appearing in the statement of Theorem \ref{thm_decay}(c) is
$B=ig'(H)$. So we trivially get that $[D,B]=0$, and by reproducing $(n-1)$ times the
previous argument we obtain that $B\in C^{n-1}(\widetilde A)$.

Summing up, the assumptions of Theorem \ref{thm_decay}(c) are satisfied for the net
$(U_j)_{j\in J}=(\e^{-itH})_{t>0}$, the set $(\ell_j)_{j\in J}=(t)_{t>0}$, the
operator $\widetilde A$, and $D=g(H)$. Thus for each
$\varphi\in g(H)^n\dom(A^n)\subset g(H)^n\dom((\widetilde A)^n)$ and
$\psi\in\dom(A^n)\subset\dom((\widetilde A)^n)$ there exists a constant
$c_{\varphi,\psi}\ge0$ such that
$$
\big|\langle\varphi,\e^{-itH}\psi\rangle_\H\big|
\le\tfrac1{t^n}\;\!c_{\varphi,\psi},\quad t>0.
$$
\end{proof}

\begin{Remark}\label{rem_other_f}
If the operators $H$ and $A$ satisfy the commutation relation $[iH,A]=f(H)$ for some
Borel function $f:\R\to\R$, then they satisfy the relation $[iH,A]=\widetilde f(H)$
for any Borel function $\widetilde f:\R\to\R$ such that
$\widetilde f|_{\sigma(H)}=f|_{\sigma(H)}$ due to functional calculus. This basic
observation will be useful in some applications in which the initial function $f$
fails to satisfy the assumptions of Proposition \ref{prop_f(H)}(b) or
\ref{prop_f(H)}(c) (see for instance Section \ref{sec_Dirac}).
\end{Remark}

%--------------------------------------------------------------------------------------
\subsection{Unitary representations with unitary generator}\label{sec_unit}
%--------------------------------------------------------------------------------------

In this section, we consider the important case where the representation is a unitary
representation $\UU:\Z\to\U(\H)$ of the additive group $\Z$. In such a case, the fact
that $\Z$ has generator $1$ implies the existence of a unitary operator $U$ in $\H$
such that $\UU(m)=U^m$ for each $m\in\Z$. One could also consider the
higher-dimensional case of a unitary representation of the additive group $\Z^d$ for
$d\ge1$. But we refrained from doing it for the sake of simplicity.

We use the notation $P_{\rm p}(U)$ (resp. $P_{\rm c}(U)$, $P_{\rm ac}(U)$) for the
projection onto the pure point (resp. continuous, absolutely continuous) subspace
$\H_{\rm p}(U)$ (resp. $\H_{\rm c}(U)$, $\H_{\rm ac}(U)$) of $U$, $E^U(\cdot)$ for the
spectral projections of $U$, and $\chi_\Theta$ for the characteristic function of a
Borel set $\Theta\subset\S^1$.

\begin{Lemma}[Properties of $D$]\label{lemma_prop_U}
Let $U$ and $A$ be a unitary and a self-adjoint operator in a Hilbert space $\H$ with
$U\in C^1(A)$, and let
$$
D_n:=\tfrac1n[A,U^n]U^{-n},\quad n\in\N^*.
$$
Then
\begin{enumerate}
\item[(a)] $\slim_{n\to\infty}D_nP_{\rm p}(U)=0$.
\item[(b)] If there exists $B\in\B(\H)$ such that $[A,U]U^{-1}-B\in\K(\H)$ and
$\slim_{n\to\infty}\tfrac1n\sum_{m=0}^{n-1}U^mBU^{-m}P_{\rm c}(U)$ exists, then
$$
\slim_{n\to\infty}D_n
=\slim_{n\to\infty}\tfrac1n\sum_{m=0}^{n-1}U^mBU^{-m}P_{\rm c}(U).
$$
\end{enumerate}
Furthermore, if $D:=\slim_{n\to\infty}D_n$ exists, then
\begin{enumerate}
\item[(c)] $[D,U^m]=0$ for all $m\in\Z$. In particular, $D$ is decomposable in the
spectral representation of $U$.
\item[(d)] $D=DP_{\rm c}(U)$. In particular, $U|_{\ker(D)^\perp}$ has purely
continuous spectrum.
\item[(e)] If $D\dom(A)\subset\dom(A)$ and
$\sum_{n\ge1}\|(D-D_n)\varphi\|^2_\H<\infty$ for all $\varphi\in\dom(A)$, then
$U|_{\ker(D)^\perp}$ has purely a.c. spectrum.
\end{enumerate}
\end{Lemma}

Point (c) implies that $\ker(D)^\perp$ is a reducing subspace for $U$. Therefore the
operator $U|_{\ker(D)^\perp}$ in points (d)-(e) is a well-defined unitary operator
(see \cite[Example~5.39(b)]{Wei_1980}).

\begin{proof}
(a) Let $\varphi\in\H$. Then $P_{\rm p}(U)\varphi=\sum_{j\ge1}\alpha_j\varphi_j$ with
$(\varphi_j)_{j\ge1}$ an orthonormal basis of $\H_{\rm p}(U)$, $\alpha_j\in\C$, and
$U\varphi_j=\theta_j\varphi_j$ for some $\theta_j\in\S^1$. Furthermore, we have for
any $n\in\N^*$ the equalities
$$
D_n
=\tfrac1n[A,U^n]U^{-n}
=\tfrac1n\sum_{m=0}^{n-1}U^m([A,U]U^{-1})U^{-m}.
$$
Therefore, we obtain
\begin{equation}\label{eq_ex_U}
\slim_{n\to\infty}D_nP_{\rm p}(U)\varphi
=\slim_{n\to\infty}\sum_{j\ge1}\alpha_j
\left(\tfrac1n\sum_{m=0}^{n-1}(U\theta_j^{-1})^m\right)[A,U]U^{-1}\varphi_j.
\end{equation}
Now
$
\|\tfrac1n\sum_{m=0}^{n-1}(U\theta_j^{-1})^m\|_{\B(\H)}\le1
$
for all $n\in\N^*$, and
$$
\slim_{n\to\infty}\tfrac1n\sum_{m=0}^{n-1}(U\theta_j^{-1})^m=E^U(\{\theta_j\})
$$
due to von Neumann's mean ergodic theorem. Therefore we can exchange the limit and the
sum in \eqref{eq_ex_U} to get
\begin{align*}
\slim_{n\to\infty}D_nP_{\rm p}(U)\varphi
=\sum_{j\ge1}\alpha_jE^U(\{\theta_j\})[A,U]U^{-1}\varphi_j
=\sum_{j\ge1}\alpha_jE^U(\{\theta_j\})[A,U]U^{-1}E^U(\{\theta_j\})\varphi_j.
\end{align*}
But $E^U(\{\theta_j\})[A,U]U^{-1}E^U(\{\theta_j\})=0$ for each $\theta_j$ due to the
virial theorem for unitary operators \cite[Prop.~2.3]{FRT_2013}. Thus we obtain that
$\slim_{n\to\infty}D_nP_{\rm p}(U)\varphi=0$, which proves the claim.

(b) Let $K:=[A,U]U^{-1}-B\in\K(\H)$. Then it follows from point (a) that
\begin{align*}
\slim_{n\to\infty}D_n
&=\slim_{n\to\infty}\tfrac1n\sum_{m=0}^{n-1}U^m([A,U]U^{-1})U^{-m}P_{\rm c}(U)\\
&=\slim_{n\to\infty}\tfrac1n\sum_{m=0}^{n-1}U^mBU^{-m}P_{\rm c}(U)
+\slim_{n\to\infty}\tfrac1n\sum_{m=0}^{n-1}U^mKU^{-m}P_{\rm c}(U).
\end{align*}
with $\slim_{n\to\infty}\tfrac1n\sum_{m=0}^{n-1}U^mKU^{-m}P_{\rm c}(U)=0$ due to
Theorem \ref{thm_RAGE-type}.

(c) The claim follows from Proposition \ref{prop_add}(b) in the case of the additive
group $X=\Z$ and the operator $A$. Indeed, if we take the proper length function
$\ell:\Z\to[0,\infty)$ given by $\ell(n):=|n|$, the unitary representation
$\UU:\Z\to\U(\H)$ given by $\UU(n):=U^n$, and the net $(x_j)_{j\in J}=(n)_{n\in\N^*}$,
then we get $\UU(m)=U^m\in C^1(A)$ for all $m\in\Z$ and
$$
\widetilde D
=\slim_{n\to\infty}\tfrac1{n-m}[A,U^{n-m}]U^{-(n-m)}
=\slim_{n\to\infty}\tfrac1n[A,U^n]U^{-n}
=\slim_{n\to\infty}D_n
=D.
$$
So all the assumptions of Proposition \ref{prop_add}(b) are verified, and thus
$[D,U^m]=0$ for all $m\in\Z$.

Finally, since $[D,U^m]=0$ for all $m\in\Z$, we have $D\chi_\Theta(U)=\chi_\Theta(U)D$
for each Borel set $\Theta\subset\S^1$, and thus $D$ is decomposable in the spectral
representation of $U$ \cite[Thm.~7.2.3(b)]{BS_1987}.

(d) The equality $D=DP_{\rm c}(U)$ follows from point (a). As a consequence, we get
that $\ker(D)^\perp\subset\H_{\rm c}(U)$, and thus the operator $U|_{\ker(D)^\perp}$
has purely continuous spectrum.

(e) Take $\varphi\in\dom(A)$. Then $\psi=D\varphi\in D\dom(A)\cap\dom(A)$, and it
follows from Theorem \ref{thm_decay}(a) that there exists $c_\psi\ge0$ such that
$$
\big|\langle\psi,U^n\psi\rangle_\H\big|
\le\|(D-D_n)\varphi\|_\H\;\!\|\psi\|_\H+\tfrac1n\;\!c_\psi,\quad n\in\N^*.
$$
So, we infer from the assumption and Cauchy-Schwarz inequality that
$\sum_{n\ge1}|\langle\psi,U^n\psi\rangle_\H|^2<\infty$, and thus that
$n\mapsto\langle\psi,U^n\psi\rangle_\H$ belongs to $\ell^2(\Z)$. Therefore,
Plancherel's theorem for the group $X=\Z$ \cite[Thm.~4.26]{Fol16} implies that
$\psi\in D\dom(A)$ belongs to the a.c. subspace $\H_{\rm ac}(U)$ of $U$. Thus
$U|_{\ker(D)^\perp}$ has purely a.c. spectrum, since $D\dom(A)$ is dense in
$\overline{D\H}=\ker(D)^\perp$ and $\H_{\rm ac}(U)$ is closed in $\H$.
\end{proof}

In the next proposition, we consider the particular case where $[A,U]=\gamma(U)$ for
some Borel function $\gamma:\S^1\to\C$. Our results in this case generalise the
results of \cite[Cor.~3.3-3.4]{Tie_2017}.

\begin{Proposition}[The case {$[A,U]=\gamma(U)$}]\label{prop_gamma(U)}
Let $U$ and $A$ be a unitary and a self-adjoint operator in a Hilbert space $\H$,
assume that $U\in C^1(A)$ with $[A,U]=\gamma(U)$ for some Borel function
$\gamma:\S^1\to\C$, and set $\eta(U):=\gamma(U)U^{-1}$. Then
\begin{enumerate}
\item[(a)] For each $\varphi\in\eta(U)\dom(A)$ and $\psi\in\dom(A)$ there exists a
constant $c_{\varphi,\psi}\ge0$ such that
$$
\big|\langle\varphi,U^n\psi\rangle_\H\big|
\le\tfrac1n\;\!c_{\varphi,\psi},\quad n\in\N^*.
$$
\item[(b)] If $\eta(U)\dom(A)\subset\dom(A)$, then $U|_{\ker(\gamma(U))^\perp}$ has
purely a.c. spectrum.
\item[(c)] Suppose that $\gamma\in C^k(\S^1)$ ($k\in\N^*$). Then for each
$\varphi\in\eta(U)^k\dom(A^k)$ and $\psi\in\dom(A^k)$ there exists a constant
$c_{\varphi,\psi}\ge0$ such that
$$
\big|\langle\varphi,U^n\psi\rangle_\H\big|
\le\tfrac1{n^k}\;\!c_{\varphi,\psi},\quad n\in\N^*.
$$
\end{enumerate}
\end{Proposition}

\begin{proof}
(a) The strong limit $D$ exists and satisfies for any $n\in\N^*$ the equalities
\begin{equation}\label{eq_same_bis}
D_n
=\tfrac1n\sum_{m=0}^{n-1}U^m([A,U]U^{-1})U^{-m}
=\eta(U)
=D.
\end{equation}
Therefore the assumptions of Theorem \ref{thm_decay}(b) are satisfied for the net
$(U_j)_{j\in J}=(U^n)_{n\in\N^*}$, the set $(\ell_j)_{j\in J}=(n)_{n\in\N^*}$, the
operator $A$, and $D=\eta(U)$. Thus for each $\varphi\in\eta(U)\dom(A)$ and
$\psi\in\dom(A)$ there exists a constant $c_{\varphi,\psi}\ge0$ such that
$$
\big|\langle\varphi,U^n\psi\rangle_\H\big|
\le\tfrac1n\;\!c_{\varphi,\psi},\quad n\in\N^*.
$$

(b) The claim follows from Lemma \ref{lemma_prop_U}(e) because $D=\eta(U)=D_n$ for all
$n\in\N^*$ and
$$
\ker(\eta(U))
=\chi_{\eta^{-1}(\{0\})}(U)\H
=\chi_{\gamma^{-1}(\{0\})}(U)\H
=\ker(\gamma(U)).
$$

(c) Since $\eta\in C(\S^1)$, we have that
\begin{equation}\label{eq_approx_bis}
\lim_{\varepsilon\searrow0}\big\|\eta(U)-\eta^\varepsilon(U)\big\|_{\B(\H)}=0
\end{equation}
with
$$
\eta^\varepsilon(U):=\sum_{m\in\Z}(\F\eta)(m)\e^{-(\varepsilon m)^2}U^m
\quad\hbox{(strong or operator norm sum)}
$$
and $\F:\ltwo(\S^1)\to\ltwo(\Z)$ the Fourier transform (see
\cite[Thm.~8.36(a)]{Fol_1999}). Using successively the fact that $D=\eta(U)$, equation
\eqref{eq_approx_bis}, the inclusion $U^m\dom(A)\subset\dom(A)$, equation
\eqref{eq_same_bis}, the relation $m(\F\eta)(m)=(\F(\id_{\S^1}\cdot\eta'))(m)$ with
$\id_{\S^1}$ the identity function on $\S^1$, and \eqref{eq_approx_bis} with $\eta$
replaced by $\id_{\S^1}\cdot\eta'$, we get for $\varphi\in\dom(A)$ the equalities
\begin{align*}
\langle A\varphi,D\varphi\rangle_\H-\langle\varphi,A\varphi\rangle_\H
&=\lim_{\varepsilon\searrow0}\sum_{m\in\Z}\overline{(\F\eta)(m)}
\e^{-(\varepsilon m)^2}\big(\langle A\varphi,U^m\varphi\rangle_\H
-\langle\varphi,U^mA\varphi\rangle_\H\big)\\
&=\lim_{\varepsilon\searrow0}\sum_{m\in\Z}\overline{(\F\eta)(m)}
\e^{-(\varepsilon m)^2}\langle\varphi,[A,U^m]\varphi\rangle_\H\\
&=\lim_{\varepsilon\searrow0}\sum_{m\in\Z}\overline{(\F\eta)(m)}
\e^{-(\varepsilon m)^2}\langle\varphi,mDU^m\varphi\rangle_\H\\
&=\lim_{\varepsilon\searrow0}\sum_{m\in\Z}
\overline{\big(\F(\id_{\S^1}\cdot\eta')\big)(m)}
\e^{-(\varepsilon m)^2}\langle\varphi,DU^m\varphi\rangle_\H\\
&=\langle\varphi,D(\id_{\S^1}\cdot\eta')(U)\varphi\rangle_\H\\
&=\langle\varphi,DU\eta'(U)\varphi\rangle_\H.
\end{align*}
Since $\eta'(U)\in\B(\H)$ due to the inclusion $\eta'\in C^{k-1}(\S^1)$, we infer that
$D\in C^1(A)$ with $[A,D]=DU\eta'(U)$. Thus, in the present case, the operator
$B\in\B(\H)$ appearing in the statement of Theorem \ref{thm_decay}(c) is
$B=U\eta'(U)$. So we trivially get that $[D,B]=0$, and by reproducing $(k-1)$ times
the previous argument we obtain that $B\in C^{k-1}(A)$.

Summing up, the assumptions of Theorem \ref{thm_decay}(c) are satisfied for the net
$(U_j)_{j\in J}=(U^n)_{n\in\N^*}$, the set $(\ell_j)_{j\in J}=(n)_{n\in\N^*}$, the
operator $A$, and $D=\eta(U)$. Thus for each $\varphi\in\eta(U)^k\dom(A^k)$ and
$\psi\in\dom(A^k)$ there exists a constant $c_{\varphi,\psi}\ge0$ such that
$$
\big|\langle\varphi,U^n\psi\rangle_\H\big|
\le\tfrac1{n^k}\;\!c_{\varphi,\psi},\quad n\in\N^*.
$$
\end{proof}

%--------------------------------------------------------------------------------------
\section{Applications}\label{sec_app}
\setcounter{equation}{0}
%--------------------------------------------------------------------------------------

In this section, we apply our results to various models in quantum mechanics and
dynamical systems. First, we consider an example of unitary representation with no
self-adjoint or unitary generator. Then we consider examples of unitary
representations having a self-adjoint generator. And finally we consider examples of
unitary representations having a unitary generator.

%--------------------------------------------------------------------------------------
\subsection{Left regular representation}\label{sec_reg}
%--------------------------------------------------------------------------------------

This example is motivated by \cite[Ex.~2.9]{RT19}. Let $X$ be a $\sigma$-compact
locally compact Hausdorff group with left Haar measure $\mu$ and proper length
function $\ell$ (see properties (L1)-(L4) in Section \ref{sec_general}). Let
$\H:=\ltwo(X,\mu)$, let $\DD\subset\H$ be the set of functions $X\to\C$ with compact
support, let $\UU:X\to\U(\H)$ be the left regular representation of $X$ on $\H$
$$
\UU(x)\varphi:=\varphi(x^{-1}\;\!\cdot),\quad x\in X,~\varphi\in\H,
$$
and let $A$ be the maximal multiplication operator by $\ell$ in $\H$
$$
A\varphi:=\ell\varphi,
\quad\varphi\in\dom(A):=\{\varphi\in\H\mid\|\ell\varphi\|_\H<\infty\}.
$$
For any $x\in X$, one has on $\DD$ the equality
$$
A\UU(x)-\UU(x)A=\big(\ell(\cdot)-\ell(x^{-1}\;\!\cdot)\big)\UU(x),
$$
with $|(\ell(\cdot)-\ell(x^{-1}\;\!\cdot))|\le\ell(x)$ due to properties (L2)-(L3).
Since $\DD$ is dense in $\dom(A)$, it follows that $\UU(x)\in C^1(A)$ with
\begin{equation}\label{eq_com_reg}
[A,\UU(x)]\UU(x)^{-1}=\ell(\cdot)-\ell(x^{-1}\;\!\cdot).
\end{equation}
In addition, we know from \cite[Ex.~2.9]{RT19} that for any net $(x_j)_{j\in J}$ in
$X$ with $x_j\to\infty$ we have
$$
D:=\slim_jD_j=\slim_j\tfrac1{\ell(x_j)}[A,\UU(x_j)]\UU(x_j)^{-1}=-1.
$$
Therefore, Theorem \ref{thm_decay}(a) applies for the net
$(U_j)_{j\in J}=(\UU(x_j))_{j\in J}$, the set
$(\ell_j)_{j\in J}=(\ell(x_j))_{j\in J}$, the operator $A$, and $D=-1$. Thus, for each
$\varphi,\psi\in\dom(A)$ there exists a constant $c_{\varphi,\psi}\ge0$ such that
$$
\big|\langle\varphi,\UU(x_j)\psi\rangle_\H\big|
\le\big\|(D-D_j)\varphi\big\|_\H\;\!\|\psi\|_\H
+\tfrac1{\ell(x_j)}\;\!c_{\varphi,\psi},\quad\ell(x_j)>0.
$$
But \eqref{eq_com_reg} and the bound
$|\ell(\cdot)-\ell(x_j^{-1}\;\!\cdot)+\ell(x_j)|\le 2\ell(\cdot)$ (which follows from
properties (L2)-(L3)) imply that
$$
\big\|(D-D_j)\varphi\big\|_\H 
=\Big\|\tfrac{\ell(\cdot)-\ell(x_j^{-1}\;\!\cdot)
+\ell(x_j)}{\ell(x_j)}\varphi\Big\|_\H
\le\tfrac2{\ell(x_j)}\|\ell\varphi\|_\H.
$$
Therefore, by setting
$\widetilde c_{\varphi,\psi}:=2\|\ell\varphi\|_\H\;\!\|\psi\|_\H+c_{\varphi,\psi}$, we
obtain that
$$
\big|\langle\varphi,\UU(x_j)\psi\rangle_\H\big|
\le\tfrac1{\ell(x_j)}\;\!\widetilde c_{\varphi,\psi},\quad\ell(x_j)>0.
$$

This estimate is similar to the estimates of Propositions \ref{prop_f(H)}(a) and
\ref{prop_gamma(U)}(a), but with the interesting difference that in general the
representation $\UU$ doesn't have neither a self-adjoint generator nor a unitary
generator. In addition, we note that one can even obtain higher order decay estimates
by carrying on the above calculations. But we refrained from presenting them here for
the sake of simplicity.

%--------------------------------------------------------------------------------------
\subsection{Schrödinger operator in $\R^n$}\label{sec_Schrod}
%--------------------------------------------------------------------------------------

In the Hilbert space $\H:=\ltwo(\R^n)$, consider the Schrödinger operator and the
generator of dilations given by
$$
H\varphi:=(-\Delta+V)\varphi
\quad\hbox{and}\quad
A\varphi:=\tfrac12(Q\cdot P+P\cdot Q),\quad\varphi\in\SS(\R^n),
$$
with $V\in\linf(\R^n,\R)$,  $Q:=(Q_1,\dots,Q_n)$ the position operator, $P:=-i\nabla$
the momentum operator and $\SS(\R^n)$ the Schwartz space on $\R^n$. Both operators are
essentially self-adjoint with closures denoted by the same symbols. If we assume that
$$
x\mapsto x\cdot(\nabla V)(x)\in\linf(\R^n,\R)
\quad\hbox{and}\quad
\lim_{|x|\to\infty}\big|2V(x)+x\cdot(\nabla V)(x)\big|=0,
$$
then we have $(H-i)^{-1}\in C^1(A)$ with $[iH,A]=2H-2V-Q\cdot\nabla V$, we have
$$
(H+i)^{-1}(2V+Q\cdot\nabla V)(H-i)^{-1}\in\K(\H)
$$
due to the standard theorem \cite[Prop.~4.1.3]{ABG_1996}, and
$$
\slim_{t\to\infty}\tfrac1t\int_0^t\d\tau\,
\e^{-i\tau H}(H+i)^{-1}2H(H-i)^{-1}\e^{i\tau H}P_{\rm c}(H)
=2H\langle H\rangle^{-2}P_{\rm c}(H).
$$
Therefore, it follows from Lemma \ref{lemma_prop_H} that
$$
D=\slim_{t\to\infty}D_t=2H\langle H\rangle^{-2}P_{\rm c}(H)
$$
and that $H|_{\ker(D)^\perp}$ has purely continuous spectrum. Moreover, Theorem
\ref{thm_decay}(a) applies for the net $(U_j)_{j\in J}=(\e^{-itH})_{t>0}$, the set
$(\ell_j)_{j\in J}=(t)_{t>0}$, the operator $\widetilde A=(H+i)^{-1}A(H-i)^{-1}$, and
$D$ as above. Therefore, for each $\varphi=D\widetilde\varphi\in D\dom(A)$ and
$\psi\in\dom(A)$ there exists a constant $c_{\varphi,\psi}\ge0$ such that
\begin{equation}\label{eq_decay_Schrod}
\big|\langle\varphi,\e^{-itH}\psi\rangle_\H\big|
\le\|(D-D_t)\widetilde\varphi\|_\H\;\!\|\psi\|_\H+\tfrac1t\;\!c_{\varphi,\psi},
\quad t>0.
\end{equation}

In this example, one cannot easily improve the decay estimate \eqref{eq_decay_Schrod}
when $V\ne0$. Indeed, in order to establish the convergence
$D_t\stackrel{\rm s}\to D$, we used Lemma \ref{lemma_prop_H}(b). But the proof of
Lemma \ref{lemma_prop_H}(b) relies on the RAGE theorem for self-adjoint operators,
whose proof relies in turn on Wiener's theorem. And as sad as it is, in general one
cannot infer an explicit rate of convergence from Wiener's theorem.

On the other hand, if $V=0$, then we have $[iH,A]=2H$, and thus the stronger decay
estimates of Proposition \ref{prop_f(H)}(c) are satisfied.

%--------------------------------------------------------------------------------------
\subsection{Dirac operator in $\R^3$}\label{sec_Dirac}
%--------------------------------------------------------------------------------------

This example is motivated by \cite[Ex.~7.7]{GLS_2016} and \cite[Sec.~7.3]{RT12_0}.
Consider in the Hilbert space $\H:=\ltwo(\R^3,\C^4)$ the Dirac operator for a
spin-$\tfrac12$ particle of mass $m>0$
$$
H\varphi:=(\alpha\cdot P+\beta m)\varphi,\quad\varphi\in\SS(\R^3,\C^4),
$$
with $\alpha:=(\alpha_1,\alpha_2,\alpha_3)$ and $\beta$ the usual $4\times4$ Dirac
matrices. Then $H$ is essentially self-adjoint (with closure denoted by the same
symbol), and we have the following result:\footnote{In \cite[Ex.~7.7]{GLS_2016}, the
authors say that it is shown in \cite{RT12_0} that $(H-i)^{-1}\in C^1(X_{\rm NW})$
with
\begin{equation}\label{eq_wrong}
[iH,X_{\rm NW}]=\sqrt{H^2-m^2}\;\!H^{-1}.
\end{equation}
This is not what is written in \cite[Sec.~7.3]{RT12_0}, and it cannot be since
$X_{\rm NW}$ on the l.h.s. of \eqref{eq_wrong} is a vector operator with three
components whereas $\sqrt{H^2-m^2}H^{-1}$ on the r.h.s. of \eqref{eq_wrong} is a
scalar operator with one component.}

\begin{Lemma}\label{lemma_dirac}
Let $X_{\rm NW}:=(X_{\rm NW,1},X_{\rm NW,2},X_{\rm NW,3})$ be the Newton-Wigner
position operator. Then the operator
$$
A\varphi:=\tfrac12(X_{\rm NW}\cdot PH^{-1}+PH^{-1}\cdot X_{\rm NW})\varphi,
\quad\varphi\in \SS(\R^3,\C^4),
$$
is essentially self-adjoint (with closure denoted by the same symbol) and
$(H-i)^{-1}\in C^1(A)$ with $[iH,A]=(H^2-m^2)H^{-2}$.
\end{Lemma}

\begin{proof}
First, we recall that $X_{\rm NW}=\F_{\rm FW}^{-1}Q\F_{\rm FW}$, with $\F_{\rm FW}$
the (unitary) Foldy-Wouthuysen transform, and that $\F_{\rm FW}$ leaves
$\SS(\R^3,\C^4)$ invariant \cite[Sec.~1.4.3]{Tha_1992}. Therefore, one gets on
$\SS(\R^3,\C^4)$ the equalities
\begin{align*}
A&=\tfrac12(X_{\rm NW}\cdot PH^{-1}+PH^{-1}\cdot X_{\rm NW})\\
&=\tfrac12\F_{\rm FW}^{-1}(Q\cdot\F_{\rm FW}PH^{-1}\F_{\rm FW}^{-1}
+\F_{\rm FW}PH^{-1}\F_{\rm FW}^{-1}\cdot Q)\F_{\rm FW}\\
&=\tfrac12\F_{\rm FW}^{-1}(Q\cdot P\beta|H|^{-1}+P\beta|H|^{-1}\cdot Q)\F_{\rm FW}
\end{align*}
with the operator within the parenthesis equal to a direct sum of operators of the
form
$$
\pm\big(Q\cdot P(P^2+m^2)^{-1/2}+P(P^2+m^2)^{-1/2}\cdot Q\big).
$$
Since each of these operators is essentially self-adjoint on $\SS(\R^3)$ due to
Nelson's lemma, one infers that $A$ is essentially self-adjoint on $\SS(\R^3,\C^4)$.

For the second claim, we recall from \cite[Sec.~7.3]{RT12_0} that
$(H-i)^{-1}\in C^1(X_{\rm NW,j})$ with $[iH,X_{\rm NW,j}]=P_jH^{-1}$. Therefore, a
calculation in the form sense on $\SS(\R^3,\C^4)$ gives
\begin{align*}
[(H-i)^{-1},A]
&=\tfrac12\sum_{j=1}^3\big([(H-i)^{-1},X_{\rm NW,j}]P_jH^{-1}
+P_jH^{-1}[(H-i)^{-1},X_{\rm NW,j}]\big)\\
&=\tfrac i2(H-i)^{-1}\sum_{j=1}^3\big([iH,X_{\rm NW,j}]P_jH^{-1}
+P_jH^{-1}[iH,X_{\rm NW,j}]\big)(H-i)^{-1}\\
&=i(H-i)^{-1}(H^2-m^2)H^{-2}(H-i)^{-1}.
\end{align*}
Since $\SS(\R^3,\C^4)$ is a core for $A$, this implies that $(H-i)^{-1}\in C^1(A)$
with $[iH,A]=(H^2-m^2)H^{-2}$.
\end{proof}

Now, take any function $f\in C^\infty(\R)$ such that $f(x)=(x^2-m^2)x^{-2}$ if
$|x|\ge m$. Then Remark \ref{rem_other_f} and Lemma \ref{lemma_dirac} imply that
$(H-i)^{-1}\in C^1(A)$ with $[iH,A]=f(H)$, and
$g(H)=f(H)\langle H\rangle^{-2}\in C^1(A)$. Thus $g(H)\dom(A)\subset\dom(A)$ and
Proposition \ref{prop_f(H)} applies. Since an application of the Foldy-Wouthuysen
transform shows that $H$ has spectrum equal to $(-\infty,-m]\cup[m,\infty)$ and
$\ker(f(H))=\{0\}$, it follows that $H$ has purely a.c. spectrum equal to
$(-\infty,-m]\cup[m,\infty)$ (which is standard knowledge) and for each
$\varphi\in g(H)^n\dom(A^n)$ and $\psi\in\dom(A^n)$ ($n\in\N^*$) there exists a
constant $c_{\varphi,\psi}\ge0$ such that
$$
\big|\langle\varphi,\e^{-itH}\psi\rangle_\H\big|
\le\tfrac1{t^n}\;\!c_{\varphi,\psi},\quad t>0.
$$

%--------------------------------------------------------------------------------------
\subsection{Quantum waveguides in $\R^n$}\label{sec_guides}
%--------------------------------------------------------------------------------------

This example is motivated by \cite{Tie_2006}. Let $\Sigma$ be a bounded open connected
set in $\R^{n-1}$ ($n\ge2$), let $\Omega:=\Sigma\times\R$ be the corresponding
waveguide with coordinates $x\equiv(\omega,x_n)$, let $-\Delta^\Sigma_{\rm D}$ be the
Dirichlet Laplacian in $\ltwo(\Sigma)$, let
$$
H_0:=1\otimes P_n^2+(-\Delta^\Sigma_{\rm D})\otimes1
\quad\hbox{in}\quad\H:=\ltwo(\Omega)\simeq\ltwo(\Sigma)\otimes\ltwo(\R),
$$
and let $H:=H_0+V$ with $V\in\linf(\Omega,\R)$ satisfying
$$
(1\otimes\langle Q_n\rangle^{1+\varepsilon})V\in\B(\H)
\quad\hbox{and}\quad
(1\otimes\langle Q_n\rangle^{1+\varepsilon})(\partial_nV)\in\B(\H)
\quad\hbox{for some $\varepsilon>0$.}
$$
Then $V\in\K(\dom(H_0),\H)$. Indeed, if we use the strongly commuting self-adjoint
operators $X_1:=1\otimes\langle P_n\rangle^2$ and
$X_2:=(-\Delta^\Sigma_{\rm D}+1)\otimes1$, we get
$$
V(H_0+2)^{-1}
=(1\otimes\langle Q_n\rangle)V
\cdot(-\Delta^\Sigma_{\rm D}+1)^{-1/2}
\otimes\langle Q_n\rangle^{-1}\langle P_n\rangle^{-1}
\cdot X_1^{1/2}X_2^{1/2}(X_1+X_2)^{-1}.
$$
But $(1\otimes\langle Q_n\rangle)V\in\B(\H)$ by assumption,
$X_1^{1/2}X_2^{1/2}(X_1+X_2)^{-1}\in\B(\H)$ due to the bound
$$
\big\|X_1^{1/2}X_2^{1/2}(X_1+X_2)^{-1}\big\|_{\B(\H)}
\le\sup_{x_1,x_2\ge1}(x_1x_2)^{1/2}(x_1+x_2)^{-1}
<\infty,
$$
and
$(-\Delta^\Sigma_{\rm D}+1)^{-1/2}
\otimes\langle Q_n\rangle^{-1}\langle P_n\rangle^{-1}\in\K(\H)
$
because $-\Delta^\Sigma_{\rm D}$ has compact resolvent and
$\langle Q_n\rangle^{-1}\langle P_n\rangle^{-1}$ is compact. Hence
$V(H_0+2)^{-1}\in\K(\H)$, and thus $V\in\K(\dom(H_0),\H)$. Therefore, the assumptions
of \cite{Tie_2006} are satisfied and the following holds true: The operator $H$ has no
singular continuous spectrum, the eigenvalues of $H$ (if any) are of finite
multiplicity and can accumulate only at the set of eigenvalues of
$-\Delta^\Sigma_{\rm D}$, and the wave operators
$W_\pm:=\slim_{t\to\pm\infty}\e^{itH}\e^{-itH_0}$ exist and are complete.

Now, set $A:=1\otimes A_n$ with $A_n$ the generator of dilations in $\ltwo(\R)$. Then
$(H-i)^{-1}\in C^1(A)$ with
$$
[iH,A]=2(1\otimes P_n^2)-(1\otimes Q_n)(\partial_nV),
$$
and an argument as above shows that the following three operators belong to $\K(\H)$:
\begin{align*}
&(H+i)^{-1}(1\otimes Q_n)(\partial_nV)(H-i)^{-1},\\
&\big((H+i)^{-1}-(H_0+i)^{-1}\big)(1\otimes P_n^2)(H-i)^{-1},\\
&(H_0+i)^{-1}(1\otimes P_n^2)\big((H-i)^{-1}-(H_0-i)^{-1}\big).
\end{align*}  
Furthermore, since $P_{\rm c}(H)=P_{\rm ac}(H)$, one has
\begin{align*}
&\slim_{\tau\to\infty}
\e^{-i\tau H}(H_0+i)^{-1}2(1\otimes P_n^2)(H_0-i)^{-1}\e^{i\tau H}P_{\rm c}(H)\\
&=2\;\!\slim_{\tau\to\infty}\e^{-i\tau H}\e^{i\tau H_0}(1\otimes P_n^2)
\langle H_0\rangle^{-2}\e^{-i\tau H_0}\e^{i\tau H}P_{\rm ac}(H)\\
&=2W_-(1\otimes P_n^2)\langle H_0\rangle^{-2}W_-^*.
\end{align*}
And since strong convergence implies strong Cesaro convergence, this implies that
\begin{align*}
&\slim_{t\to\infty}\tfrac1t\int_0^t\d\tau\,
\e^{-i\tau H}(H_0+i)^{-1}2(1\otimes P_n^2)(H_0-i)^{-1}\e^{i\tau H}P_{\rm c}(H)\\
&=2W_-(1\otimes P_n^2)\langle H_0\rangle^{-2}W_-^*.
\end{align*}
Thus it follows from Lemma \ref{lemma_prop_H}(b) that
$$
D=\slim_{t\to\infty}D_t=2W_-(1\otimes P_n^2)\langle H_0\rangle^{-2}W_-^*.
$$
So, Theorem \ref{thm_decay}(a) applies for the net $(U_j)_{j\in J}=(\e^{-itH})_{t>0}$,
the set $(\ell_j)_{j\in J}=(t)_{t>0}$, the operator
$\widetilde A=(H+i)^{-1}A(H-i)^{-1}$, and $D$ as above. Therefore, for each
$\varphi=D\widetilde\varphi\in D\dom(A)$ and $\psi\in\dom(A)$ there exists a constant
$c_{\varphi,\psi}\ge0$ such that
$$
\big|\langle\varphi,\e^{-itH}\psi\rangle_\H\big|
\le\|(D-D_t)\widetilde\varphi\|_\H\;\!\|\psi\|_\H+\tfrac1t\;\!c_{\varphi,\psi},
\quad t>0.
$$

Of course, one obtains stronger decay estimates when $V=0$. Indeed, one gets in such a
case $[iH,A]=2(1\otimes P_n^2)$ and $D=D_t=2(1\otimes P_n^2)\langle H\rangle^{-2}$ for
all $t>0$, from which one can infer the estimates of Theorem \ref{thm_decay}(c).

%--------------------------------------------------------------------------------------
\subsection{Stark Hamiltonian in $\R^n$}\label{sec_Stark}
%--------------------------------------------------------------------------------------

This example is motivated by \cite[Sec.~7.3]{GLS_2016} and \cite[Sec.~7.1]{RT12_0}.
Let $\H:=\ltwo(\R^n)$ and take a unit vector $v\in\R^n$. Then the Stark Hamiltonian
$H$ with electric field along $v$ and the operator $A$ given by
$$
(H\varphi)(x):=-(\Delta\varphi)(x)+(v\cdot x)\varphi(x)
\quad\hbox{and}\quad
A\varphi:=i(v\cdot\nabla)\varphi,\quad\varphi\in\SS(\R^n),~x\in\R^n.
$$
are essentielly self-adjoint (with closures denoted by the same symbols), and
$(H-i)^{-1}\in C^\infty(A)$ with $[iH,A]=1$. Thus Proposition \ref{prop_f(H)} applies
with $f(H)=1$, $g(H)=\langle H\rangle^{-2}\in C^\infty(A)$ and $\ker(f(H))=\{0\}$. It
follows that $H$ has purely a.c. spectrum (which is standard knowledge) and for each
$\varphi\in\langle H\rangle^{-2n}\dom(A^n)$ and $\psi\in\dom(A^n)$ ($n\in\N^*$) there
exists a constant $c_{\varphi,\psi}\ge0$ such that
$$
\big|\langle\varphi,\e^{-itH}\psi\rangle_\H\big|
\le\tfrac1{t^n}\;\!c_{\varphi,\psi},\quad t>0.
$$
This estimate is similar to the result of \cite[Sec.~7.3]{GLS_2016} (see
\cite[Thm.~6.1]{GLS_2016}).

%--------------------------------------------------------------------------------------
\subsection{Fractional Laplacian in $\R^n$}\label{sec_frac}
%--------------------------------------------------------------------------------------

This example is motivated by \cite[Sec.~7.5]{GLS_2016}. For $s\in(0,2)$, let
$H:=(-\Delta)^{s/2}$ be the fractional Laplacian in the Hilbert space
$\H:=\ltwo(\R^n)$ and let $A$ be the generator of dilations in $\H$. Then
$(H-i)^{-1}\in C^\infty(A)$ with $[iH,A]=sH$. Thus Proposition \ref{prop_f(H)} applies
with $f(H)=sH$, $g(H)=sH\langle H\rangle^{-2}\in C^\infty(A)$ and
$\ker(f(H))=\ker(H)$. Since an application of the Fourier transform shows that $H$ has
spectrum equal to $[0,\infty)$ and $\ker(H)=\{0\}$, it follows that $H$ has purely
a.c. spectrum equal to $[0,\infty)$ (which is standard knowledge) and for each
$\varphi\in g(H)^n\dom(A^n)$ and $\psi\in\dom(A^n)$ ($n\in\N^*$) there exists a
constant $c_{\varphi,\psi}\ge0$ such that
$$
\big|\langle\varphi,\e^{-itH}\psi\rangle_\H\big|
\le\tfrac1{t^n}\;\!c_{\varphi,\psi},\quad t>0.
$$
This estimate improves the result of \cite[Sec.~7.5]{GLS_2016} (see
\cite[Thm.~6.3]{GLS_2016}).

%--------------------------------------------------------------------------------------
\subsection{Horocycle flow}\label{sec_horo}
%--------------------------------------------------------------------------------------

This example is motivated by \cite{Tie_2012,Tie_2017}, but see also
\cite{Tie_2015,Tie_2017_2}. Let $\Sigma$ be a  finite volume Riemann surface of genus
$\ge2$ and let $M:=T^1\Sigma$ be the unit tangent bundle of $\Sigma$. The $3$-manifold
$M$ carries a probability measure $\mu_\Omega$, induced by a canonical volume form
$\Omega$, which is preserved by two distinguished one-parameter groups of
diffeomorphisms\;\!: the horocycle flow $F_1:=(F_{1,t})_{t\in\R}$ and the geodesic
flow $F_2:=(F_{2,t})_{t\in\R}$. Each flow admits a self-adjoint generator $H_j$ in
$\H:=\ltwo(M,\mu_\Omega)$ essentially self-adjoint on $C^\infty_{\rm c}(M)$ and given
by
$$
H_j\varphi:=-i\L_{X_j}\varphi,\quad\varphi\in C^\infty_{\rm c}(M),
$$
with $X_j$ the divergence-free vector field associated to $F_j$ and $\L_{X_j}$ the
corresponding Lie derivative. Moreover, one has $(H_1-i)^{-1}\in C^\infty(H_2)$ with
$[iH_1,H_2]=H_1$, see \cite[Sec.~3]{Tie_2012}. Thus Proposition \ref{prop_f(H)}
applies with $f(H_1)=H_1$, $g(H_1)=H_1\langle H_1\rangle^{-2}\in C^\infty(H_2)$ and
$\ker(f(H_1))=\ker(H_1)$. It follows that $H_1|_{\ker(H_1)^\perp}$ has purely a.c.
spectrum (which is standard knowledge) and for each $\varphi\in g(H_1)^n\dom(H_2^n)$
and $\psi\in\dom(H_2^n)$ ($n\in\N^*$) there exists a constant $c_{\varphi,\psi}\ge0$
such that
$$
\big|\langle\varphi,\psi\circ F_{1,t}\rangle_{\H}\big|
=\big|\langle\varphi,\e^{-itH_1}\psi\rangle_\H\big|
\le\tfrac1{t^n}\;\!c_{\varphi,\psi},\quad t>0.
$$

Roughly, this estimate shows that if the vectors $\varphi,\psi$ are $n$ times
differentiable along the geodesic flow, then their correlation along the horocycle
flow decays as $\tfrac1{t^n}$. It thus provides a new version of the polynomial decay
of correlations for the horocycle flow on the unit tangent bundle of a finite volume
surface of constant negative curvature. And its (rather short) proof did not use the
identification of $M$ with a homogeneous space $\Gamma\setminus{\rm PSL}(2,\R)$ and
the representation theory associated to it as is customary (see for instance
\cite{Moo_1987}).

Now, assume that $\Sigma$ is compact and consider a $C^1$ vector field proportional to
$X_1$, that is, $fX_1$ with $f\in C^1(M,(0,\infty))$. The vector field $fX_1$ admits a
complete flow $\widetilde F_1:=(\widetilde F_{1,t})_{t\in\R}$ uniquely ergodic with
respect to the measure
$\widetilde\mu_\Omega:=\frac{f^{-1}\mu_\Omega}{\int_Mf^{-1}\d\mu_\Omega}$ and with
generator $H:=fH_1$ essentially self-adjoint on
$C^1(M)\subset\H:=\ltwo(M,f^{-1}\mu_\Omega)$. Furthermore, the following holds true
\cite[Ex.~4.8]{Tie_2017}: The operator $A:=f^{1/2}H_2f^{-1/2}$ is essentially
self-adjoint on $C^1(M)\subset\H$, one has $(H-i)^{-1}\in C^1(A)$ with
$$
[iH,A]=H\xi+\xi H,\quad\xi:=\tfrac12-\tfrac12f^{-1}\L_{X_2}(f),
$$
and $D=\slim_{t\to\infty}D_t=H\langle H\rangle^{-2}$. Thus Theorem \ref{thm_decay}(a)
applies for the net $(U_j)_{j\in J}=(\e^{-itH})_{t>0}$, the set
$(\ell_j)_{j\in J}=(t)_{t>0}$, the operator $\widetilde A=(H+i)^{-1}A(H-i)^{-1}$, and
$D=H\langle H\rangle^{-2}$. Therefore, for each
$\varphi=D\widetilde\varphi\in D\dom(A)$ and $\psi\in\dom(A)$ there exists a constant
$c_{\varphi,\psi}\ge0$ such that
\begin{equation}\label{eq_decay_horo}
\big|\langle\varphi,\psi\circ\widetilde F_{1,t}
\rangle_{\H}\big|
=\big|\langle\varphi,\e^{-itH}\psi\rangle_\H\big|
\le\|(D-D_t)\widetilde\varphi\|_\H\;\!\|\psi\|_\H+\tfrac1t\;\!c_{\varphi,\psi},
\quad t>0.
\end{equation}

In this case, one cannot easily improve the decay estimate \eqref{eq_decay_horo} with
the tools of this paper. Indeed, in order to establish the convergence
$D_t\stackrel{\rm s}\to D$, one uses the ergodic theorem for uniquely ergodic flows
(see \cite{Tie_2017}) which does not come with an explicit rate of convergence. We
refer to \cite[Thm.~19]{FU12} for a more quantitative decay estimate obtained using
the representation theory of ${\rm SL}(2,\R)$.

%--------------------------------------------------------------------------------------
\subsection{Adjacency matrices}\label{sec_graphs}
%--------------------------------------------------------------------------------------

This example is motivated by \cite{MRT_2007} and \cite[Ex.~4.7]{Tie_2017}. Let
$(X,\sim)$ be a graph of finite degree, with symmetric relation $\sim$, and with no
multiple edges or loops. Then the adjacency matrix of $(X,\sim)$ is the bounded
self-adjoint operator in the Hilbert space $\H:=\ell^2(X)$ given by
$$
(H\varphi)(x):=\sum_{y\sim x}\varphi(y),\quad\varphi\in\H,~x\in X.
$$
A directed graph $(X,<)$ subjacent to $(X,\sim)$ is the graph $(X,\sim)$ together with
a relation $<$ on $X$ such that, for each $x,y\in X$, $x\sim y$ is equivalent to $x<y$
or $y<x$ (and one cannot have both $x<y$ and $y<x$). When using drawings, we draw an
arrow $x\leftarrow y$ if $x<y$. For each $x\in X$, we define the sets
$N^-(x):=\{y\in X\mid x<y\}$ and $N^+(x):=\{y\in X\mid y<x\}$. Then a directed graph
$(X,<)$ is called admissible if \cite[Def.~5.1]{MRT_2007}:
\begin{enumerate}
\item[(i)] each closed path in $X$ has the same number of positively and negatively
oriented edges,
\item[(ii)] for each $x,y\in X$, one has
$\#\{N^-(x)\cap N^-(y)\}=\#\{N^+(x)\cap N^+(y)\}$.
\end{enumerate}

\begin{figure}[h]
\centering
\includegraphics[width=340pt]{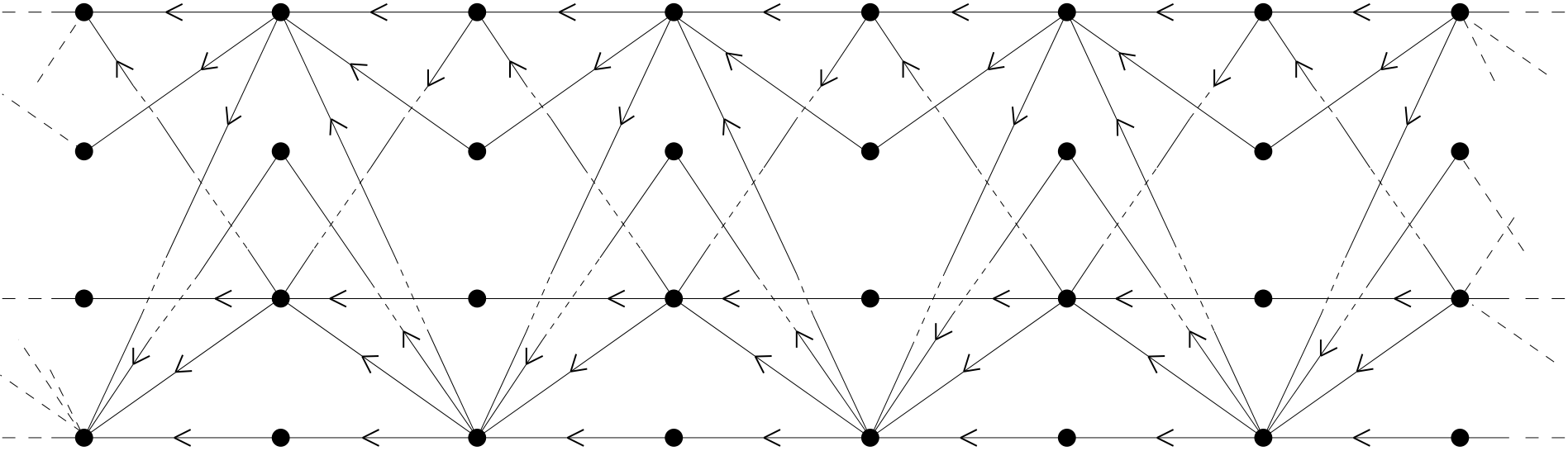}
\caption{Example of admissible graph}
\end{figure}

In the case of admissible graphs, it is shown in \cite[Secs.~3-5]{MRT_2007} that there
exist a self-adjoint operator $A$ and a bounded self-adjoint operator $K$ in $\H$ such
that $H\in C^1(A)$ with $[iH,A]=K^2$, $K^2\in C^1(A)$ with $[iK^2,A]=-2HK^2$, and
$[K,H]=0$. Thus Lemma \ref{lemma_prop_H} applies with
$D=D_t=K^2\langle H\rangle^{-2}\in C^1(A)$ for all $t>0$ and $\ker(D)=\ker(K)$. So
$H|_{\ker(K)^\perp}$ has purely a.c. spectrum, as was first proved in
\cite[Thm.~1.1]{MRT_2007} with a more complicated method.

Now, since $H,K^2\in C^1(A)$, we have $(H\pm i)^{-1}\dom(A)\subset\dom(A)$ and
$K^2\dom(A)\subset\dom(A)$. Thus a calculation using the above commutation relations
and the operator $\widetilde A:=(H+i)^{-1}A(H-i)^{-1}$ gives on $\dom(A)$:
\begin{align*}
[\widetilde A,D]
&=[(H+i)^{-1}A(H-i)^{-1},K^2\langle H\rangle^{-2}]\\
&=K^2(H+i)^{-1}[A,\langle H\rangle^{-2}](H-i)^{-1}
+(H+i)^{-1}[A,K^2](H-i)^{-1}\langle H\rangle^{-2}\\
&=K^2(H+i)^{-1}(-2iK^2H\langle H\rangle^{-4})(H-i)^{-1}
+(H+i)^{-1}(-2iHK^2)(H-i)^{-1}\langle H\rangle^{-2}\\
&=-2iK^2H\langle H\rangle^{-4}(K^2\langle H\rangle^{-2}+1)\\
&=-2iDH\langle H\rangle^{-2}(D+1).
\end{align*}
Since $DH\langle H\rangle^{-2}(D+1)\in\B(\H)$ and $\dom(A)$ is a core for
$\widetilde A$, we infer that $D\in C^1(\widetilde A)$ with
$[\widetilde A,D]=-2iDH\langle H\rangle^{-2}(D+1)$. So, in the present case, the
operator $B\in\B(\H)$ appearing in the statement of Theorem \ref{thm_decay}(c) is
$B=-2iH\langle H\rangle^{-2}(D+1)$ and it satisfies $[D,B]=0$. Furthermore, a repeated
use of the information gathered so far shows that each factor appearing in the
expression for $B$ belongs to $C^{(n-1)}(\widetilde A)$ for any $n\in\N^*$. Thus
$B\in C^{(n-1)}(\widetilde A)$ for any $n\in\N^*$. It follows that the assumptions of
Theorem \ref{thm_decay}(c) are satisfied for the net
$(U_j)_{j\in J}=(\e^{-itH})_{t>0}$, the set $(\ell_j)_{j\in J}=(t)_{t>0}$, the
operator $\widetilde A$, and $D=K^2\langle H\rangle^{-2}$. Thus for each
$\varphi\in D^n\dom(A^n)$ and $\psi\in\dom(A^n)$ ($n\in\N^*$) there exists a constant
$c_{\varphi,\psi}\ge0$ such that
$$
\big|\langle\varphi,\e^{-itH}\psi\rangle_\H\big|
\le\tfrac1{t^n}\;\!c_{\varphi,\psi},\quad t>0.
$$

%--------------------------------------------------------------------------------------
\subsection{Jacobi matrices}\label{sec_Jacobi}
%--------------------------------------------------------------------------------------

This example is motivated by \cite{Sah_2008}. Let $\H:=\ell^2(\N^*)$, let
$\ell^2_0\subset\H$ be the set of functions $\N^*\to\C$ with compact support, and let
$H$ be the Jacobi matrix
$$
(H\varphi)(n):=a_{n-1}\varphi(n-1)+b_n\varphi(n)+a_n\varphi(n+1),
\quad\varphi\in\ell^2_0,~n\in\N^*,~\varphi(0):=0,
$$
with coefficients
$$
a_n:=n^\alpha+\eta_n,\quad b_n:=\lambda(n^\alpha+(n-1)^\alpha)+\beta_n,
\quad\alpha>0,~\lambda,\eta_n,\beta_n\in\R.
$$
For any sequence $(r_n)_{n\in\N^*}$ and $k\in\N^*$, define by induction
$(\partial^{k+1}r)_n:=(\partial\partial^kr)_n$ with $(\partial r)_n:=r_{n+1}-r_n$, set
$$
\beta_n':=\beta_n-\lambda(\eta_n+\eta_{n-1}),\quad n\in\N^*,~\eta_0:=0,
$$
and assume the following:

\begin{Assumption}
Suppose that $|\lambda|<1$, $\alpha\in(0,1)$, and assume that the sequences with
elements
$$
\tfrac{\eta_n}{n^\alpha},\quad
(\partial\eta)_n,\quad
n^{1-2\alpha}(\partial\eta)_n,\quad
n^{1-\alpha}(\partial^2\eta)_n,\quad
\beta_n',\quad
n^{1-\alpha}(\partial\beta')_n,\quad
(n\in\N^*)
$$
vanish as $n\to\infty$.
\end{Assumption}

Then it is shown in \cite[Thms~1.1-1.4]{Sah_2008} that $H$ is essentially self-adjoint
(with closure denoted by the same symbol), that $\sigma_{\rm ess}(H)=\R$, and that the
eigenvalues of $H$ (if any) are of finite multiplicity and can accumulate only at
$\pm\infty$. Furthermore, it is shown in the proof of \cite[Cor.~8.1]{Sah_2008} that
there exists a self-adjoint operator $A$ in $\H$ such that $(H-i)^{-1}\in C^1(A)$ with
$$
[iH,A]=4(1-\lambda^2)(1-\alpha)+K,\quad K\in\K(\dom(H),\H).
$$
Since $(H+i)^{-1}K(H-i)^{-1}\in\K(\H)$ and
\begin{align*}
&\slim_{t\to\infty}\tfrac1t\int_0^t\d\tau\,
\e^{-i\tau H}(H+i)^{-1}4(1-\lambda^2)(1-\alpha)(H-i)^{-1}\e^{i\tau H}P_{\rm c}(H)\\
&=4(1-\lambda^2)(1-\alpha)\langle H\rangle^{-2}P_{\rm c}(H),
\end{align*}
it follows from Lemma \ref{lemma_prop_H} that
$$
D=\slim_{t\to\infty}D_t=4(1-\lambda^2)(1-\alpha)\langle H\rangle^{-2}P_{\rm c}(H)
$$
and that $H|_{\ker(D)^\perp}=H|_{P_{\rm c}(H)\H}$ has (trivially) purely continuous
spectrum. Moreover, Theorem \ref{thm_decay}(a) applies for the net
$(U_j)_{j\in J}=(\e^{-itH})_{t>0}$, the set $(\ell_j)_{j\in J}=(t)_{t>0}$, the
operator $\widetilde A=(H+i)^{-1}A(H-i)^{-1}$, and $D$ as above. Therefore, for each
$\varphi=D\widetilde\varphi\in D\dom(A)$ and $\psi\in\dom(A)$ there exists a constant
$c_{\varphi,\psi}\ge0$ such that
$$
\big|\langle\varphi,\e^{-itH}\psi\rangle_\H\big|
\le\|(D-D_t)\widetilde\varphi\|_\H\;\!\|\psi\|_\H+\tfrac1t\;\!c_{\varphi,\psi},
\quad t>0.
$$

One can obtain stronger decay estimates in particular cases. For instance, it is shown
in \cite[Appx.~A]{Sah_2008} that for certain choices of $a_n$ and $b_n$ one gets the
relation $[iH,A]=aH+b$ with $(a,b)\in\R^2\setminus\{(0,0)\}$, from which one can infer
the estimates of Theorem \ref{thm_decay}(c).

%--------------------------------------------------------------------------------------
\subsection{Schrödinger operators on Fock spaces}\label{sec_Fock}
%--------------------------------------------------------------------------------------

This example is motivated by \cite{GG_2005}. Let $U$ be an isometry in a Hilbert space
$\H$ and let $N$ be a number operator for $U$, that is, a self-adjoint operator in
$\H$ such that $U^*\dom(N)\subset\dom(N)$ and $UNU^*=N-1$ on $\dom(N)$. Number
operators do not always exist, but if $U$ is completely non-unitary, namely if
$\slim_{n\to\infty}(U^*)^n=0$, then $N$ exists and is unique. Next, consider the
bounded self-adjoint operators $H:=\re(U)$ and $S:=\im(U)$, and set
$$
A\varphi:=\tfrac12(SN+NS)\varphi,\quad\varphi\in\dom(N).
$$
Then it is shown in \cite[Secs.~2-3]{GG_2005} that $A$ is essentially self-adjoint
(with closure denoted by the same symbol) and that $H\in C^\infty(A)$ with
$[iH,A]=1-H^2$. Thus Proposition \ref{prop_f(H)} applies with $f(H)=1-H^2$,
$g(H)=(1-H^2)\langle H\rangle^{-2}\in C^\infty(A)$ and $\ker(f(H))=E^H(\{-1,1\})\H$.
It follows that $H$ has purely a.c. spectrum in $(-1,1)$ and for each
$\varphi\in g(H)^n\dom(A^n)$ and $\psi\in\dom(A^n)$ ($n\in\N^*$) there exists a
constant $c_{\varphi,\psi}\ge0$ such that
$$
\big|\langle\varphi,\e^{-itH}\psi\rangle_\H\big|
\le\tfrac1{t^n}\;\!c_{\varphi,\psi},\quad t>0.
$$

Operators like $H$ (and additive perturbations of it) appear naturally in the
framework of Fock spaces and their applications to Schrödinger operators on trees.
Indeed, let $\h$ be a complex Hilbert space and
$\H:=\bigoplus_{n=0}^\infty\h^{\otimes n}$ the complete Fock space associated to it
(with $\h^{\otimes0}:=\C$ and $\h^{\otimes n}:=\{0\}$ if $n<0$). For any $u\in\h$ with
$\|u\|_\h=1$, let $U\in\B(\H)$ act on the sector $\h^{\otimes n}$ as
$$
U(h_1\otimes\dots\otimes h_n):=h_1\otimes\dots\otimes h_n\otimes u,
\quad h_1\otimes\dots\otimes h_n\in\h^{\otimes n}.
$$
Then $U$ is completely non-unitary, the operator $H=\re(U)$ acts on the sector
$\h^{\otimes n}$ as
$$
\begin{cases}
H(h_1\otimes\dots\otimes h_n)
=\tfrac12\big(h_1\otimes\dots\otimes h_n\otimes u 
+h_1\otimes\dots\otimes h_{n-1}\langle h_n,u\rangle_\h\big) & \hbox{if $n\ge1$,}\\
Hh=\tfrac12hu & \hbox{if $h\in\h^{\otimes 0}=\C$,}
\end{cases}
$$
and the number operator $N$ for $U$ can be explicitly described (see
\cite[Sec.~4]{GG_2005} for more details).

%--------------------------------------------------------------------------------------
\subsection{Multiplication by $\lambda$ in $\ltwo(\R_+,\d\mu)$}\label{sec_lambda}
%--------------------------------------------------------------------------------------

This example is motivated by \cite[Sec.~7.6]{GLS_2016}. Let $H$ be the maximal
multiplication operator by the variable $\lambda\in\R_+$ in the Hilbert space
$\H:=\ltwo(\R_+,\d\mu)$, with $\d\mu:=h\;\!\d\lambda$ and $h\in C^1(\R_+,\R_+)$, and
assume that the function
$$
q:\R_+\to\R,~~\lambda\mapsto\tfrac{\lambda h'(\lambda)}{h(\lambda)}+1,
$$
is bounded. Then the operator
$$
A\varphi:=-\tfrac i4(2H\varphi'+q\varphi),\quad\varphi\in C^\infty_{\rm c}(\R_+),
$$
is essentially self-adjoint (with closure denoted by the same symbol) and
$(H-i)^{-1}\in C^\infty(A)$ with
$[iH,A]=-\tfrac12H$.\footnote{In \cite[Sec.~7.6]{GLS_2016}, there is a small mistake
in the calculation of the commutator $[iH,A]$. This is why our operator $A$ and
commutator $[iH,A]$ slightly differ from the ones appearing in
\cite[Sec.~7.6]{GLS_2016}.} Thus Proposition \ref{prop_f(H)} applies with
$f(H)=-\tfrac12H$, $g(H)=-\tfrac12H\langle H\rangle^{-2}\in C^\infty(A)$ and
$\ker(f(H))=\ker(H)$. Furthermore, since
$$
\mu\big(\{\lambda\in [0,\infty)\mid\lambda=0\}\big)=0,
$$
we have that $\ker(H)=\{0\}$. It follows that $H$ has purely a.c. spectrum equal to
$[0,\infty)$ and for each $\varphi\in g(H)^n\dom(A^n)$ and $\psi\in\dom(A^n)$
($n\in\N^*$) there exists a constant $c_{\varphi,\psi}\ge0$ such that
$$
\big|\langle\varphi,\e^{-itH}\psi\rangle_\H\big|
\le\tfrac1{t^n}\;\!c_{\varphi,\psi},\quad t>0.
$$
This estimate improves the result of \cite[Sec.~7.6]{GLS_2016} (see
\cite[Thm.~6.3]{GLS_2016}).

%--------------------------------------------------------------------------------------
\subsection{$H=-\partial_{xx}+\partial_{yy}$ in $\R^2$}\label{sec_partial}
%--------------------------------------------------------------------------------------

This example is motivated by \cite[Sec.~7.2]{GLS_2016}. Consider in the Hilbert space
$\H:=\ltwo(\R^2)$ the partial differential operator
$$
H\varphi:=(-\partial_{xx}+\partial_{yy})\varphi,\quad\varphi\in\SS(\R^2),
$$
which is essentially self-adjoint (with closure denoted by the same symbol), and $A$
the generator of dilations in $\H$. Then we have $(H-i)^{-1}\in C^\infty(A)$ with
$[iH,A]=2H$. Thus Proposition \ref{prop_f(H)} applies with $f(H)=2H$,
$g(H)=2H\langle H\rangle^{-2}\in C^\infty(A)$ and $\ker(f(H))=\ker(H)$. Since an
application of the Fourier transform shows that $H$ has spectrum equal to $\R$ and
$\ker(H)=\{0\}$, it follows that $H$ has purely a.c. spectrum equal to $\R$, and for
each $\varphi\in g(H)^n\dom(A^n)$ and $\psi\in\dom(A^n)$ ($n\in\N^*$) there exists a
constant $c_{\varphi,\psi}\ge0$ such that
$$
\big|\langle\varphi,\e^{-itH}\psi\rangle_\H\big|
\le\tfrac1{t^n}\;\!c_{\varphi,\psi},\quad t>0.
$$
This estimate improves the result of \cite[Sec.~7.2]{GLS_2016} (see
\cite[Thm.~6.3]{GLS_2016}).

%--------------------------------------------------------------------------------------
\subsection{$H=-X^{2-s}\Delta-\Delta X^{2-s}$ in $\R_+$}\label{sec_homo}
%--------------------------------------------------------------------------------------

This example is motivated by \cite[Sec.~7.4]{GLS_2016}. Let $X$ be the maximal
multiplication operator by the variable $x\in\R_+$ in the Hilbert space $\H:=\ltwo(\R_+)$ and
let $A$ be the generator of dilations in $\H$. Then the operator
$$
H\varphi:=(-X^{2-s}\Delta-\Delta X^{2-s})\varphi,
\quad\varphi\in C^\infty_{\rm c}(\R_+),~s\in(0,2),
$$
is essentially self-adjoint (with closure denoted by the same symbol) and
$(H-i)^{-1}\in C^\infty(A)$ with $[iH,A]=sH$. In particular, one obtains that
$\e^{-itA}H=\e^{st}H\e^{-itA}$ for all $t\in\R$ which implies that $\ker(H)=\{0\}$,
since otherwise $\ker(H)$ would be a nontrivial invariant subspace of the irreducible
representation $\R\ni t\mapsto\e^{-itA}\in\U(\H)$. Thus Proposition \ref{prop_f(H)}
applies with $f(H)=sH$, $g(H)=sH\langle H\rangle^{-2}\in C^\infty(A)$ and
$\ker(f(H))=\ker(H)=\{0\}$. It follows that $H$ has purely a.c. spectrum, and for each
$\varphi\in g(H)^n\dom(A^n)$ and $\psi\in\dom(A^n)$ ($n\in\N^*$) there exists a
constant $c_{\varphi,\psi}\ge0$ such that
$$
\big|\langle\varphi,\e^{-itH}\psi\rangle_\H\big|
\le\tfrac1{t^n}\;\!c_{\varphi,\psi},\quad t>0.
$$
This estimate improves the result of \cite[Sec.~7.4]{GLS_2016} (see
\cite[Thm.~6.3]{GLS_2016}).

%--------------------------------------------------------------------------------------
\subsection{Quantum walks on $\Z$}\label{sec_Z}
%--------------------------------------------------------------------------------------

This example is motivated by \cite{RST_2018,RST_2019}. Consider a quantum walk on $\Z$
with evolution operator $U:=SC$ in the Hilbert space $\H:=\ell^2(\Z,\C^2)$, where $S$
is the shift operator defined as
$$
(S\varphi)(x)
:=\left(\begin{smallmatrix}
\varphi^{(0)}(x+1)\\
\varphi^{(1)}(x-1)
\end{smallmatrix}\right),
\quad
\varphi
=\left(\begin{smallmatrix}
\varphi^{(0)}\\
\varphi^{(1)}
\end{smallmatrix}\right)\in\H,~x\in\Z,
$$
and $C$ the coin operator defined as
$$
(C\varphi)(x):=C(x)\varphi(x),\quad\varphi\in\H,~x\in\Z,~C(x)\in\U(2).
$$
Assume that $C$ is anisotropic, namely, converging with short-range rate to an
asymptotic coin on the left and to an asymptotic coin on the right:

\begin{Assumption}
There exist $C_\ell,C_{\sf r}\in\U(2)$, $\kappa_\ell,\kappa_{\sf r}>0$, and
$\varepsilon_\ell,\varepsilon_{\sf r}>0$ such that
\begin{align*}
&\|C(x)-C_\ell\|_{\B(\C^2)}\le\kappa_\ell\;\!|x|^{-1-\varepsilon_\ell},\quad x<0,\\
&\|C(x)-C_{\sf r}\|_{\B(\C^2)}\le\kappa_{\sf r}\;\!|x|^{-1-\varepsilon_{\sf r}},
\quad x>0,
\end{align*}
where the indexes $\ell$ and ${\sf r}$ stand for ``left" and ``right".
\end{Assumption}

Under this assumption, it is shown in \cite[Sec.~4]{RST_2018} that $U$ has no singular
continuous spectrum and that the eigenvalues of $U$ (if any) are of finite
multiplicity and can accumulate only at a finite set of threshold values. Furthermore,
there exist a self-adjoint operator $A$ in $\H$, a unitary operator $U_0$ in an
auxiliary Hilbert space $\H_0$, a self-adjoint operator $A_0$ in $\H_0$ and
$J\in\B(\H_0,\H)$ such that $U\in C^1(A)$, $U_0\in C^1(A_0)$ and
$$
[A_0,U_0]U_0^{-1}=V_0^2
\quad\hbox{and}\quad
[A,U]U^{-1}-J[A_0,U_0]U_0^{-1}J^*\in\K(\H)
$$
with $V_0\in\B(\H_0)$ an asymptotic velocity operator for $U_0$ satisfying
$[V_0,U_0]=0$. Since $P_{\rm c}(U)=P_{\rm ac}(U)$, one then infers from
\cite[Sec.~3]{RST_2019} that
$$
\slim_{n\to\infty}U^nJV_0^2J^*U^{-n}P_{\rm c}(U)
=\left(\slim_{n\to\infty}U^nJU_0^{-n}P_{\rm ac}(U_0)\right)
V_0^2\left(\slim_{n\to\infty}U_0^nJ^*U^{-n}P_{\rm ac}(U)\right)
=W_-V_0^2W_-^*
$$
with
\begin{equation}\label{eq_wave}
W_\pm:=\slim_{n\to\pm\infty}U^{-n}JU_0^nP_{\rm ac}(U_0)
\end{equation}
the wave operators for the triple $(U,U_0,J)$. Since strong convergence implies strong
Cesaro convergence, we obtain that
$$
\slim_{n\to\infty}\tfrac1n\sum_{m=0}^{n-1}U^mJV_0^2J^*U^{-m}P_{\rm c}(U)
=W_-V_0^2W_-^*,
$$
and thus it follows from Lemma \ref{lemma_prop_U}(b) that
$$
D=\slim_{n\to\infty}D_n=W_-V_0^2W_-^*.
$$
So, Theorem \ref{thm_decay}(a) applies for the net $(U_j)_{j\in J}=(U^n)_{n\in\N^*}$,
the set $(\ell_j)_{j\in J}=(n)_{n\in\N^*}$, the operator $A$, and $D=W_-V_0^2W_-^*$.
Therefore, for each $\varphi=D\widetilde\varphi\in D\dom(A)$ and $\psi\in\dom(A)$
there exists a constant $c_{\varphi,\psi}\ge0$ such that
\begin{equation}\label{eq_decay_Z}
\big|\langle\varphi,U^n\psi\rangle_\H\big|
\le\|(D-D_n)\widetilde\varphi\|_\H\;\!\|\psi\|_\H+\tfrac1n\;\!c_{\varphi,\psi},
\quad n\in\N^*.
\end{equation}

In this example, one cannot easily improve the decay estimate \eqref{eq_decay_Z}.
Indeed, in order to establish the convergence $D_n\stackrel{\rm s}\to D$, we used
Lemma \ref{lemma_prop_U}(b). But the proof of Lemma \ref{lemma_prop_U}(b) relies on
the RAGE theorem \ref{thm_RAGE}, whose proof relies in turn on a discrete version of
Wiener's theorem. And as in the continuous case, in general one cannot infer an
explicit rate of convergence from Wiener's theorem. In addition, we used the strong
limits \eqref{eq_wave} which do not come with an explicit rate of convergence either.

%--------------------------------------------------------------------------------------
\subsection{Quantum walks on trees}\label{sec_tree}
%--------------------------------------------------------------------------------------

This example is motivated by \cite{Tie_2021}. Let $\Tau$ be a homogeneous tree of odd
degree $d\ge3$, that is, a finitely generated group $\Tau$ with generators
$a_1,\dots,a_d$, identity $e$, and presentation
$\Tau:=\langle a_1,\dots,a_d\mid a_1^2=\dots=a_d^2=e\rangle$.
\begin{center}
\includegraphics[width=220pt]{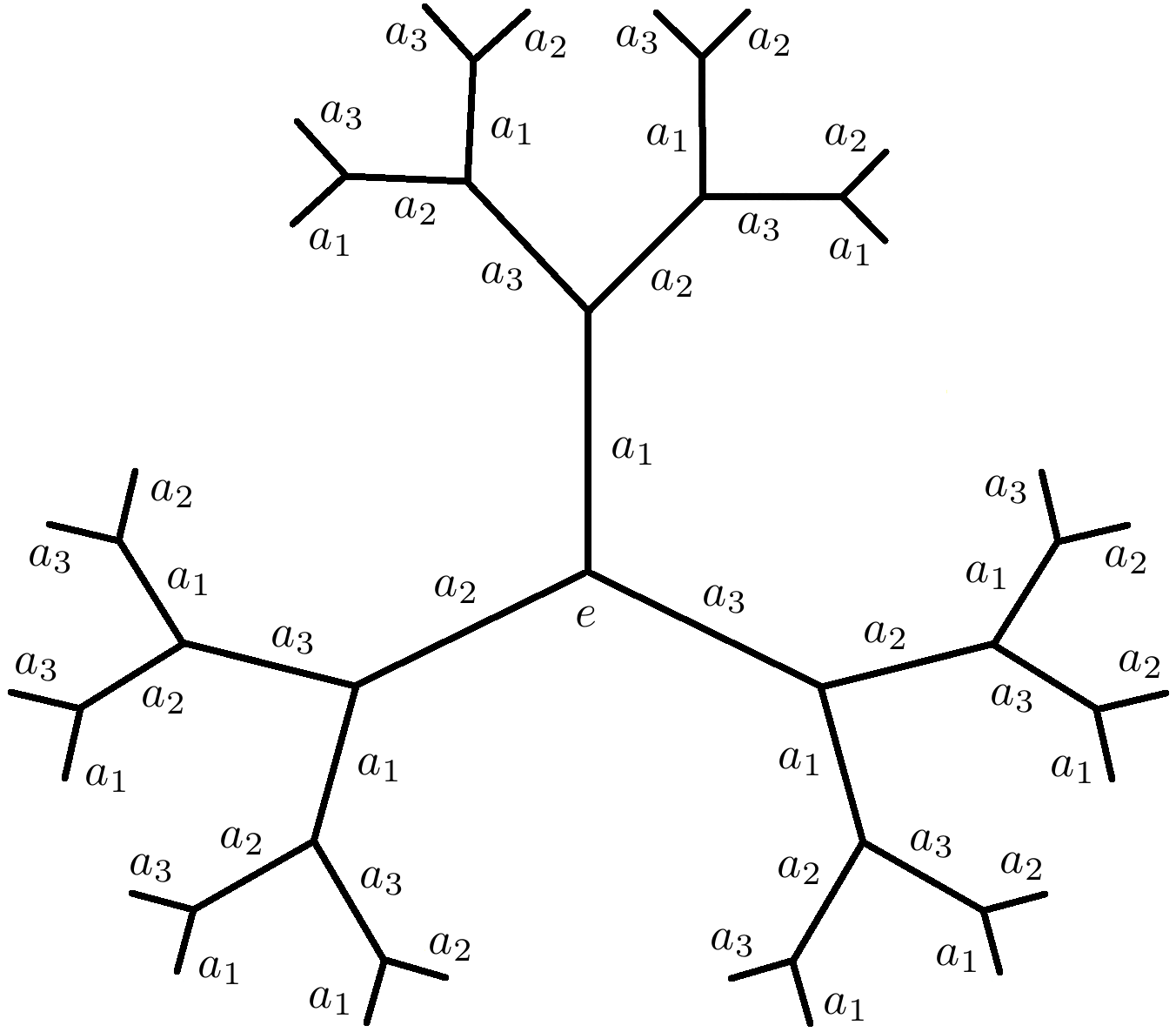}\\
\medskip
{\small Homogeneous tree $\Tau$ of degree $d=3$}
\end{center}
Using the word length $|\cdot|$ on $\Tau$, we define sets of even/odd elements of
$\Tau$
$$
\Tau_{\sf e}:=\{x\in\Tau\mid|x|\in2\N\}
\quad\hbox{and}\quad
\Tau_{\sf o}:=\{x\in\Tau\mid|x|\in2\N+1\}
$$
with corresponding characteristic functions $\chi_{\sf e}$ and $\chi_{\sf o}$. Then we
consider a quantum walk on $\Tau$ with evolution operator $U:=SC$ in the Hilbert space
$\H:=\ell^2(\Tau,\C^d)$, where $S$ is the shift operator defined as
\begin{gather*}
S:=
\left(\begin{smallmatrix}
S_{1+1,1+2} &&&\\
& S_{2+1,2+2} && \mbox{\large$0$}\\
\mbox{\large$0$} && \ddots &\\
&&& S_{d+1,d+2}\\
\end{smallmatrix}\right),
\quad S_{d,d+1}:=S_{d,1},~S_{d+1,d+2}:=S_{1,2},\\
S_{i,j}f:=\chi_{\sf e}f(\;\!\cdot\;\!a_i)+\chi_{\sf o}f(\;\!\cdot\;\!a_j),
\quad i,j\in\{1,\dots,d\},~f\in\ell^2(\Tau),
\end{gather*}
and $C$ the coin operator defined as
$$
(C\varphi)(x):=C(x)\varphi(x),\quad\varphi\in\H,~x\in\Tau,~C(x)\in\U(d).
$$

Assume that $C$ is anisotropic, namely, converging with short-range rate to a diagonal
asymptotic coin on each main branch of $\Tau$:

\begin{Assumption}
For $i=1,\dots,d$, there exist a diagonal matrix $C_i\in\U(d)$ and $\varepsilon_i>0$
such that
$$
\|C(x)-C_i\|_{\B(\C^d)}\le{\rm Const.}\;\!\langle x\rangle^{-(1+\varepsilon_i)},
\quad x\in\Tau_i,
$$
where $\Tau_i:=\{x\in\Tau\mid|a_ix|=|x|-1\}$.
\end{Assumption}

Under this assumption, it is shown in \cite[Sec.~5]{Tie_2021} that the spectrum of $U$
covers the whole unit circle and is purely absolutely continuous, outside possibly a
finite set where $U$ may have eigenvalues of finite multiplicity. Furthermore, there
exist a self-adjoint operator $A$ in $\H$, a unitary operator $U_0$ in an auxiliary
Hilbert space $\H_0$, a self-adjoint operator $A_0$ in $\H_0$ and $J\in\B(\H_0,\H)$
such that $U\in C^1(A)$, $U_0\in C^\infty(A_0)$ and
$$
JJ^*=1_\H,
\quad[A_0,U_0]U_0^{-1}=2,
\quad[A,U]U^{-1}-J[A_0,U_0]U_0^{-1}J^*\in\K(\H).
$$
It then follows from Lemma \ref{lemma_prop_U}(b) that
$$
D=\slim_{n\to\infty}D_n
=\slim_{n\to\infty}\tfrac1n\sum_{m=0}^{n-1}
U^m(J[A_0,U_0]U_0^{-1}J^*)U^{-m}P_{\rm c}(U)
=2P_{\rm c}(U).
$$
Thus, Theorem \ref{thm_decay}(a) applies for the net
$(U_j)_{j\in J}=(U^n)_{n\in\N^*}$, the set $(\ell_j)_{j\in J}=(n)_{n\in\N^*}$, the
operator $A$, and $D=2P_{\rm c}(U)$. So for each
$\varphi=D\widetilde\varphi\in D\dom(A)$ and $\psi\in\dom(A)$ there exists a constant
$c_{\varphi,\psi}\ge0$ such that
\begin{equation}\label{eq_decay_tree}
\big|\langle\varphi,U^n\psi\rangle_\H\big|
\le\|(D-D_n)\widetilde\varphi\|_\H\;\!\|\psi\|_\H+\tfrac1n\;\!c_{\varphi,\psi},
\quad n\in\N^*.
\end{equation}
As in the example of the previous section, one cannot easily improve the decay
estimate \eqref{eq_decay_tree}.

%--------------------------------------------------------------------------------------
\subsection{Skew products}\label{sec_skew}
%--------------------------------------------------------------------------------------

This example is motivated by \cite{Tie_2018}, but see also
\cite{CT_2016,Tie_2015,Tie_2015_2,Tie_2017}. Let $X$ be a smooth compact second
countable Hausdorff manifold with Borel probability measure $\mu_X$, and let
$(F_t)_{t\in\R}$ be a $C^1$ measure-preserving flow on $X$. The operators
$V_t:\ltwo(X,\mu_X)\to\ltwo(X,\mu_X)$ given by $V_t\varphi:=\varphi\circ F_t$ define a
strongly continuous one-parameter unitary group with self-adjoint generator $H$ in
$\ltwo(X,\mu_X)$ essentially self-adjoint on $C^1(X)$ and given by
$$
H\varphi:=i\L_Y\varphi,\quad\varphi\in C^1(X),
$$
with $Y$ the $C^0$ vector field associated to $F$ and $\L_Y$ the corresponding Lie
derivative.

Let $G$ be a compact Lie group with identity $e_G$, normalised Haar measure $\mu_G$,
and Lie algebra $\g$. Then any measurable function $\phi:X\to G$ induces a measurable
cocycle $X\times\Z\ni(x,n)\mapsto\phi^{(n)}(x)\in G$ over $F_1$ given by
$$
\phi^{(n)}(x):=
\begin{cases}
\phi(x)(\phi\circ F_1)(x)\cdots(\phi\circ F_{n-1})(x) & \hbox{if $n\ge1$}\\
\hfil e_G & \hbox{if $n=0$}\\
\hfil(\phi^{(-n)}\circ F_n)(x)^{-1} & \hbox{if $n\le-1$.}
\end{cases}
$$
The skew product $T_\phi$ defined by
$$
T_\phi:X\times G\to X\times G,~~(x,g)\mapsto\big(F_1(x),g\;\!\phi(x)\big),
$$
is an automorphism of $(X\times G,\mu_X\otimes\mu_G)$, and the corresponding Koopman
operator
$$
U_\phi\psi
:=\psi\circ T_\phi,\quad\psi\in\H:=\ltwo(X\times G,\mu_X\otimes\mu_G),
$$
is a unitary operator in $\H$.

Let $\widehat G$ be the set of (equivalence classes of) finite-dimensional irreducible
unitary representations of $G$. Then each $\pi\in\widehat G$ is a $C^\infty$ group
homomorphism from $G$ to the unitary group $\U(d_\pi)$ of degree $d_\pi:=\dim(\pi)$,
and Peter-Weyl's theorem implies that the set of all matrix elements
$\{\pi_{jk}\}_{j,k=1}^{d_\pi}$ of all $\pi\in\widehat G$ forms an orthogonal basis of
$\ltwo(G,\mu_G)$. Accordingly, one has the orthogonal decomposition
\begin{equation}\label{eq_decompo}
\H=\bigoplus_{\pi\in\widehat G}\,\bigoplus_{j=1}^{d_\pi}\H^{(\pi)}_j,
\quad\H^{(\pi)}_j:=\bigoplus_{k=1}^{d_\pi}\ltwo(X,\mu_X)\otimes\{\pi_{jk}\},
\end{equation}
and $U_\phi$ is reduced by the decomposition \eqref{eq_decompo}, with restriction
$U_{\phi,\pi,j}:=U_\phi\big|_{\H^{(\pi)}_j}$ given by
$$
U_{\phi,\pi,j}\sum_{k=1}^{d_\pi}\varphi_k\otimes\pi_{jk}
=\sum_{k,\ell=1}^{d_\pi}(\varphi_k\circ F_1)(\pi_{\ell k}\circ\phi)\otimes\pi_{j\ell},
\quad\varphi_k\in\ltwo(X,\mu_X).
$$
Furthermore, the following holds true \cite[Sec.~3]{Tie_2018}: The operator $A$ given
by
$$
A\sum_{k=1}^{d_\pi}\varphi_k\otimes\pi_{jk}
:=\sum_{k=1}^{d_\pi}H\varphi_k\otimes\pi_{jk},
\quad\varphi_k\in C^1(X),
$$
is essentially self-adjoint in $\H^{(\pi)}_j$ (with closure denoted by the same
symbol). If $\L_Y(\pi\circ\phi)\in\linf(X,\B(\C^{d_\pi}))$, then
$U_{\phi,\pi,j}\in C^1(A)$ with
$$
[A,U_{\phi,\pi,j}]=iM_{\pi\circ\phi}\;\!U_{\phi,\pi,j}
$$
and $M_{\pi\circ\phi}$ the bounded matrix-valued multiplication operator in
$\H^{(\pi)}_j$ given by
$$
M_{\pi\circ\phi}\sum_{k=1}^{d_\pi}\varphi_k\otimes\pi_{jk}
:=\sum_{k,\ell=1}^{d_\pi}\big(\L_Y(\pi\circ\phi)\cdot(\pi\circ\phi)^{-1}\big)_{k\ell}
\;\!\varphi_\ell\otimes\pi_{jk},\quad\varphi_k\in\ltwo(X,\mu_X).
$$
Finally, if $\L_Y\phi$ exists $\mu_X$-almost everywhere and $M_\phi\in\ltwo(X,\g)$,
then
$$
D=\slim_{n\to\infty}D_n=i(\d\pi)_{e_G}\big((P_\phi M_\phi)(\cdot)\big)
$$
where $(\d\pi)_{e_G}\big((P_\phi M_\phi)(\cdot)\big)$ is a bounded operator in
$\H^{(\pi)}_j$ that can be interpreted as a matrix-valued degree of the cocycle
$\pi\circ\phi:X\to\pi(G)$ (see \cite[Rem.~3.12]{Tie_2018} for more details). Thus,
Theorem \ref{thm_decay}(a) applies for the net
$(U_j)_{j\in J}=(U_{\phi,\pi,j}^n)_{n\in\N^*}$, the set
$(\ell_j)_{j\in J}=(n)_{n\in\N^*}$, the operator $A$, and
$D=i(\d\pi)_{e_G}((P_\phi M_\phi)(\cdot))$. So for each
$\varphi=D\widetilde\varphi\in D\dom(A)$ and $\psi\in\dom(A)$ there exists a constant
$c_{\varphi,\psi}\ge0$ such that
\begin{equation}\label{eq_decay_skew}
\big|\big\langle\varphi,U_{\phi,\pi,j}^n\psi\big\rangle_{\H^{(\pi)}_j}\big|
\le\|(D-D_n)\widetilde\varphi\|_{\H^{(\pi)}_j}\|\psi\|_{\H^{(\pi)}_j}
+\tfrac1n\;\!c_{\varphi,\psi},\quad n\in\N^*.
\end{equation}

In this example, one cannot easily improve the decay estimate \eqref{eq_decay_skew}.
Indeed, the convergence $D_n\stackrel{\rm s}\to D$ follows from
\cite[Lemma~3.3]{Tie_2018}, whose proof relies on Birkhoff's pointwise ergodic theorem
for Banach-valued functions. And without additional information, one cannot infer an
explicit rate of convergence from that theorem. That being said, one can exhibit
various examples where \eqref{eq_decay_skew} is satisfied, and one can even prove that
$U_\phi$ has purely a.c. spectrum in appropriate subspaces of $\H$. We refer to
\cite[Sec.~4]{Tie_2018} for more details.

%--------------------------------------------------------------------------------------
\appendix
%--------------------------------------------------------------------------------------

%--------------------------------------------------------------------------------------
\section{Commutators and regularity classes}\label{sec_comm}
\setcounter{equation}{0}
\renewcommand{\theequation}{A.\arabic{equation}}
%--------------------------------------------------------------------------------------

In this appendix, we recall the definitions of commutators of operators and regularity
classes associated with them that we use in this work. We refer to Chapters 5-6 of the
monograph \cite{ABG_1996} for more details.

Let $A$ be a self-adjoint operator in a Hilbert space $\H$ with domain $\dom(A)$, and
take a bounded operator $S\in\B(\H)$. For any $k\in\N$, we say that $S$ belongs to
$C^k(A)$, with notation $S\in C^k(A)$, if the map
\begin{equation}\label{eq_Heisen}
\R\ni t\mapsto\e^{-itA}S\e^{itA}\in\B(\H)
\end{equation}
is strongly of class $C^k$. The sets $C^k(A)\subset\B(\H)$ satisfy the inclusions
$$
C^\infty(A):=\cap_{k\in\N}C^k(A)\subset\cdots\subset C^2(A)\subset C^1(A)
\subset C^0(A)=\B(\H).
$$
In the case $k=1$, one has $S\in C^1(A)$ if and only if the quadratic form
$$
\dom(A)\ni\varphi\mapsto\langle\varphi,iSA\varphi\rangle_\H
-\langle A\varphi,iS\varphi\rangle_\H\in\C
$$
is continuous for the topology induced by $\H$ on $\dom(A)$. We denote by $[iS,A]$ the
bounded operator associated with the continuous extension of this form, or
equivalently the strong derivative of the map \eqref{eq_Heisen} at $t=0$.

If $H$ is a self-adjoint operator in $\H$ with domain $\dom(H)$ and spectrum
$\sigma(H)$, we say that $H$ is of class $C^k(A)$ if $(H-z)^{-1}\in C^k(A)$ for some
$z\in\C\setminus\sigma(H)$. In particular, $H$ is of class $C^1(A)$ if and only if the
quadratic form
$$
\dom(A)\ni\varphi\mapsto\langle\varphi,(H-z)^{-1}A\varphi\rangle_\H
-\langle A\varphi,(H-z)^{-1}\varphi\rangle_\H\in\C
$$
extends continuously to a bounded form with corresponding bounded operator denoted by
$[(H-z)^{-1},A]\in\B(\H)$. In such a case, the set $\dom(H)\cap\dom(A)$ is a core for
$H$ and the quadratic form
$$
\dom(H)\cap\dom(A)\ni\varphi\mapsto\langle H\varphi,A\varphi\rangle
-\langle A\varphi,H\varphi\rangle\in\C
$$
is continuous in the graph norm topology of $\dom(H)$ \cite[Thm.~6.2.10(a)]{ABG_1996}.
This form then extends uniquely to a continuous quadratic form on $\dom(H)$ which can
be identified with a continuous operator $[H,A]$ from $\dom(H)$ to the adjoint space
$\dom(H)^*$. In addition, the following relation holds in $\B(\H)$
\cite[Thm.~6.2.10(b)]{ABG_1996}:
$$
[(H-z)^{-1},A]=-(H-z)^{-1}[H,A](H-z)^{-1}.
$$

%--------------------------------------------------------------------------------------
\section{RAGE-type theorem for unitary operators}\label{sec_RAGE}
\setcounter{equation}{0}
\renewcommand{\theequation}{B.\arabic{equation}}
%--------------------------------------------------------------------------------------

In this appendix, we give the proof of a RAGE-type theorem for unitary operators that
we use in Section \ref{sec_unit}. The theorem is surely well-known, but since we did
not find it in this form in the literature we present its proof for completeness.

We start by recalling the usual RAGE theorem for a unitary operator $U$. As in the
previous sections, we use the notation $P_{\rm p}(U)$ for the projection onto the pure
point subspace $\H_{\rm p}(U)$ of $U$ and $P_{\rm c}(U)$ for the projection onto the
continuous subspace $\H_{\rm c}(U)$ of $U$.

\begin{Theorem}[RAGE theorem, page 320 of \cite{Simon_2015}]\label{thm_RAGE}
Let $U$ be a unitary operator in a Hilbert space $\H$ and $K\in\K(\H)$. Then
$$
\lim_{n\to\infty}\tfrac1n\sum_{m=0}^{n-1}\|KU^{-m}\varphi\|_\H^2
=\|KP_{\rm p}(U)\varphi\|_\H^2\quad\hbox{for all $\varphi\in\H$.}
$$
\end{Theorem}

\begin{Theorem}\label{thm_RAGE-type}
Let $U$ be a unitary operator in a Hilbert space $\H$ and $K\in\K(\H)$. Then
$$
\slim_{n\to\infty}\tfrac1n\sum_{m=0}^{n-1}U^mKU^{-m}
=\sum_{\theta\in\scriptsize\{\hbox{eigenvalues of $U$}\}}
E^U(\{\theta\})KE^U(\{\theta\}).
$$
\end{Theorem}

\begin{proof}
We mimic the proof of the analogous theorem in the self-adjoint case
\cite[Thm.~5.9]{Tes_2014}. Any $\varphi \in \H$ admits an orthogonal decomposition
$\varphi=\varphi_{\rm p}+\varphi_{\rm c}$ with $\varphi_{\rm p}\in\H_{\rm p}(U)$ and
$\varphi_{\rm c}\in\H_{\rm c}(U)$. For the component $\varphi_{\rm c}$, we get from
the Cauchy-Schwarz inequality and Theorem \ref{thm_RAGE} that
$$
\lim_{n\to\infty}\left\|\tfrac1n\sum_{m=0}^{n-1}U^mKU^{-m}\varphi_{\rm c}\right\|_\H
\le\lim_{n\to\infty}\left(\sum_{m=0}^{n-1}\tfrac1n\right)^{1/2}
\left(\tfrac1n\sum_{m=0}^{n-1}\big\|KU^{-m}\varphi_{\rm c}\big\|_\H^2\right)^{1/2}
=0.
$$
For the component $\varphi_{\rm p}$, we write
$\varphi_{\rm p}=\sum_{j\ge1}\alpha_j\varphi_j$ with $(\varphi_j)_{j\ge1}$ an
orthonormal basis of $\H_{\rm p}(U)$, $\alpha_j\in\C$, and
$U\varphi_j=\theta_j\varphi_j$ for some $\theta_j\in\S^1$. Then we get
\begin{equation}\label{eq_exchange}
\slim_{n\to\infty}\tfrac1n\sum_{m=0}^{n-1}U^mKU^{-m}\varphi_{\rm p}
=\slim_{n\to\infty}\sum_{j\ge1}\alpha_j
\left(\tfrac1n\sum_{m=0}^{n-1}(U\theta_j^{-1})^m\right)K\varphi_j.
\end{equation}
Now we have
$
\|\tfrac1n\sum_{m=0}^{n-1}(U\theta_j^{-1})^m\|_{\B(\H)}\le1
$
for all $n\in\N^*$, and
$$
\slim_{n\to\infty}\tfrac1n\sum_{m=0}^{n-1}(U\theta_j^{-1})^m=E^U(\{\theta_j\})
$$
due to von Neumann's mean ergodic theorem. Therefore we can exchange the limit and the
sum in \eqref{eq_exchange} to obtain
$$
\slim_{n\to\infty}\tfrac1n\sum_{m=0}^{n-1}U^mKU^{-m}\varphi_{\rm p}
=\sum_{j\ge1} \alpha_j E^U(\{\theta_j\}) K\varphi_j
=\sum_{\theta\in\scriptsize\{\hbox{eigenvalues of $U$}\}}
E^U(\{\theta\})KE^U(\{\theta\})\varphi,
$$
as desired.
\end{proof}

%--------------------------------------------------------------------------------------
%\bibliography{../bibliographie/bibliographie}
%--------------------------------------------------------------------------------------

\def\cprime{$'$} \def\polhk#1{\setbox0=\hbox{#1}{\ooalign{\hidewidth
  \lower1.5ex\hbox{`}\hidewidth\crcr\unhbox0}}}
  \def\polhk#1{\setbox0=\hbox{#1}{\ooalign{\hidewidth
  \lower1.5ex\hbox{`}\hidewidth\crcr\unhbox0}}}
  \def\polhk#1{\setbox0=\hbox{#1}{\ooalign{\hidewidth
  \lower1.5ex\hbox{`}\hidewidth\crcr\unhbox0}}} \def\cprime{$'$}
  \def\cprime{$'$} \def\polhk#1{\setbox0=\hbox{#1}{\ooalign{\hidewidth
  \lower1.5ex\hbox{`}\hidewidth\crcr\unhbox0}}}
  \def\polhk#1{\setbox0=\hbox{#1}{\ooalign{\hidewidth
  \lower1.5ex\hbox{`}\hidewidth\crcr\unhbox0}}} \def\cprime{$'$}
  \def\cprime{$'$} \def\cprime{$'$}

%--------------------------------------------------------------------------------------

\end{document}